\documentclass[a4paper,fleqn]{cas-dc}

\usepackage[numbers]{natbib}
\usepackage[capitalise]{cleveref}
\usepackage{marginnote}
\usepackage[dvipsnames]{xcolor}
\usepackage{courier}
\usepackage{amsmath}
\usepackage{xfrac}
\usepackage{amsthm}
\usepackage{amssymb}
\usepackage{mathtools}
\usepackage{units}
\usepackage{booktabs}
\usepackage{cleveref}
\usepackage{bm}
\usepackage[acronym]{glossaries}
\usepackage{graphicx}
\usepackage{subcaption}
\usepackage{caption}
\usepackage{float}
\usepackage{stfloats}   
\usepackage{placeins}
\graphicspath{{Figures/}}

\newif\ifmargincomments 
\margincommentstrue

\newtheorem{lemma}{Lemma}
\newtheorem{theorem}{Theorem}

\newtheorem{prob}{Problem}

\theoremstyle{definition}

\theoremstyle{remark}

\newif\ifextendedversion 
\extendedversionfalse

\ifmargincomments
\else

\fi

\newif\ifjournal
\journaltrue     

\def\tsc#1{\csdef{#1}{\textsc{\lowercase{#1}}\xspace}}
\tsc{WGM}
\tsc{QE}

\begin{document}
\let\WriteBookmarks\relax
\def\floatpagepagefraction{1}
\def\textpagefraction{.001}

\shorttitle{Power Management Policy for Hybrid Aero Engines}    

\shortauthors{F. Pak, U. Perra, F. Willems, T.  Hofman, M. Salazar}  

\title [mode = title]{An Optimal Power Management Policy for Hydrogen-based Hybrid Aero Engines}  

\tnotemark[1] 

\tnotetext[1]{Research funded by the European Union (ID 101138488) and by the UK Research and Innovation (ID 10106893) as a part of the FlyECO project. Views expressed are those of the author(s) only and do not necessarily reflect those of the funding entities. Neither the European Union nor the granting authority can be held responsible for them.~\cite{FLYECO}} 

%

\author[1]{Faezeh Pak}

\cormark[1]

\fnmark[1]

\ead{f.pak@tue.nl}


\credit{Conceptualization, Methodology, Data curation, Software implementation, Writing}

\affiliation[1]{organization={Eindhoven University of Technology},
            addressline={Groene Loper 3}, 
            city={Eindhoven},
            postcode={5612AE}, 
            state={North Brabant},
            country={The Netherlands}}

\author[1]{Uto Perra}

\fnmark[2]

\ead{u.perra@tue.nl}


\credit{Conceptualization, Data curation, Software implementation, Writing, Review and editing}


\author[1]{Frank Willems}

\fnmark[3]

\ead{F.P.T.Willems@tue.nl}


\credit{Co-author, Proof-read, Writing and review, Supervision}


\author[1]{Theo Hofman}

\fnmark[4]

\ead{T.Hofman@tue.nl}


\credit{Co-author, Proof-read, Review and editing, Supervision}


\author[1]{Mauro Salazar}

\fnmark[5]

\ead{m.r.u.salazar@tue.nl}


\credit{Methodology, Writing, Proof-read, Review and editing, Supervision}


\cortext[1]{Corresponding author}



\begin{abstract}
	This paper presents a power management policy for a hydrogen-based hybrid aero engine combining a gas turbine and a solid oxide fuel cell (SOFC).
	Specifically, we first identify a quadratic quasi-steady-state model of the propulsion system and formulate the minimum-fuel optimal control problem as a function of the power split between gas turbine and SOFC that captures the interconnections between the components and accounts for their operational limits.
	Second, leveraging the Karush-Kuhn-Tucker optimality conditions and partial convexity and monotonicity model properties, we compute the globally optimal steady-state power split for the different phases of the flight in closed form.
	Finally, we verify this power management policy with a high-fidelity integrated static model 
	across different flight phases, revealing in less than $1.5 \%$ normalized root mean square error in power allocation and less than $0.7 \%$ in predicted fuel consumption.
	Our results show that the optimal power management policy can be translated into a heuristic control law requesting the highest SOFC power that does not exceed its maximum operating temperature, ultimately paving the way for minimal-effort on-board implementations.
\end{abstract}


\begin{highlights}
\item We develop an optimization problem for a hydrogen hybrid aero engine by using curvature-constrained regression to ensure uniqueness of the solution.
\item We formulate an optimal power management problem using the tractable formulation and solved analytically using Karush-Kuhn-Tucker condition. 
\item A closed-form analytical power management is derived. It enables real-time implementation without iterative numerical calculations.
\item The validation against high fidelity models demonstrates less than $1.5 \%$ normalized root mean square error in power allocation and lower than $0.7 \%$ normalized root mean square error in fuel consumption.
\end{highlights}


\begin{keywords}
Power management \sep Optimal control \sep Solid oxide fuel cell \sep Hybrid propulsion system
\end{keywords}
\maketitle
\section{Introduction}
The aviation industry has committed to becoming climate neutral by 2050, partly by adopting \ifjournal a range of novel electrical \fi aerospace propulsion concepts  
\ifjournal
such as more-electric aircraft (MEA)~\citep{cheng2020}, hybrid-electric aircraft (HEA) \citep{reid2024}, and all-electric aircraft (AEA)~\citep{xia2013}. In this classification, MEA enhances subsystem electrification without fundamentally altering propulsion, e.g., hydraulic and pneumatic systems~\citep{li2021}, and AEA concepts eliminate combustion entirely. Although the AEAs can reach higher efficiency~\citep{xia2013}, their main drawbacks are their increasing mass due to the current battery technology and long charging time~\citep{li2021}. On the other hand, HEAs combine the conventional gas turbine engines with electrical power sources in order to enhance the efficiency, power density and feasibility. In such configurations, determining the optimal power allocation between subsystems becomes a key challenge~\citep{liang2023}.
%
\else
.
\fi

To contribute to the transition toward clean aviation, this study investigates a power management policy for a hybrid electric propulsion system as illustrated in~\cref{fig:scheme}, consisting of a hydrogen gas turbine (GT) and a solid oxide fuel cell (SOFC) that are thermodynamically coupled, to power a propeller via a gas turbine and an electric motor, respectively. 
\ifjournal 
The SOFC/GT configuration has only one power source, which is hydrogen here, still it lies within HEA concepts since it represents a combined-cycle implementation~\citep{geuther2020,BenitoEtAl2025}.
\fi
This complex aero engine requires a high-level power management policy to coordinate the components' operation and minimize hydrogen consumption while realizing the requested output power.  Although the analysis of such a system can be structured either in a steady‑state operation or a transient operation, the current work focuses exclusively on the steady‑state behavior of the propulsion system during take off, top of climb, and cruise, laying the foundation for future investigations into transient dynamics. 

\emph{Related Literature:}
Research in this area spans three streams: (i) control systems for hybrid and conventional aircraft engines, (ii) coordinated control of SOFC/GT integrated systems, and (iii) minimum-fuel strategies in automotive propulsion. 
In conventional aero-engines, optimization methods such as the maximum principle for the Mayer optimal control problem~\citep{CotsEtAl2018}, dynamic programming for power tracking~\citep{SunEtAl2021}, and iterative Hamilton-Jacobi-Bellman reformulations of minimum fuel optimal power tracking during transient operation~\citep{LiuEtAL2024} have been used to improve engine performance. 

For hybrid SOFC/GT propulsion, PID-based turbine speed control~\citep{ZaccariaEtAl2016} and neural-network multi-loop controllers with optimal tracking~\citep{RestrepoEtAl2016} have been developed for power tracking rather than fuel-optimal control. The studies on SOFC/GT configuration in different applications, including in power plants~\citep{KameswaranEtAl2007}, in the marine industry~\citep{MiEtAl2025}, vehicles~\citep{dimitrovaEtAl2017}, and in aviation~\citep{MotaponEtAl2013}, mainly rely on multi‑loop coordinated control architectures in which inner‑loop controllers regulate the local subsystem dynamics, while an outer-loop controller tracks the required power. 
\begin{figure*}[t]
	\begin{center}
		\includegraphics[width=15cm]{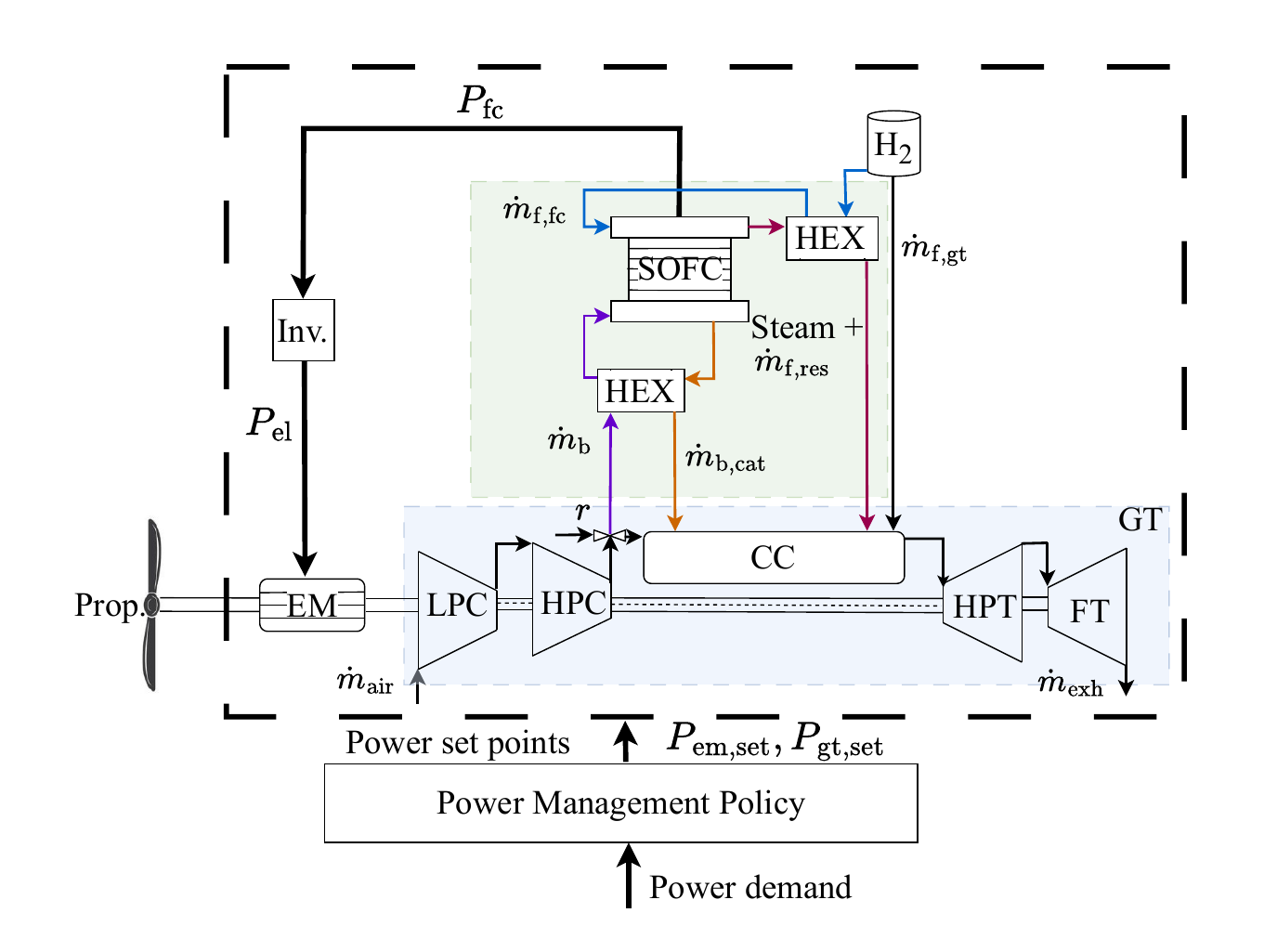}
		\ifjournal
		\else
		\vspace{-20pt}
   	    \fi
		\caption{Simplified schematic view of the aero engine, including a GT, the SOFC, and an electrical motor (EM)~\citep{BenitoEtAl2025}. In this figure, $\dot{m}_{\mathrm{f,gt}}$ and $\dot{m}_{\mathrm{f,fc}}$ are the fuel flow directed to the GT and SOFC, respectively. $\dot{m}_{\mathrm{b}}$ is the bleed air flow, which is guided from the compressor to the SOFC system after warm-up, and $\dot{m}_{\mathrm{f,res}}$ is the SOFC off-gas hydrogen. $P_{\mathrm{gt}}$, $P_{\mathrm{fc}}$, and $P_{\mathrm{em}}$ are GT output power, SOFC electrical power, and EM power, respectively.}
		\label{fig:scheme}
	\end{center}
\ifjournal
\else
\vspace{-20pt}
\fi
\end{figure*}
In these approaches, power allocation is typically generated by PID controller~\citep{ChenEtAl2017}, while the delay in SOFC temperature is not considered, or by optimal approaches using dynamic programming and NLP frameworks for power tracking and maximizing the efficiency~\citep{KameswaranEtAl2007} or via equilibrium-optimizer algorithms embedded in neural-network-based coordination scheme~\citep{MiEtAl2025} to avoid local minima in naval warship by fixing the fuel utilization of the SOFC and manipulating the current through changing the fuel flow rate. These controllers focus on meeting power demand rather than computing fuel-optimal supervisory policies for aviation propulsion. 

Automotive research offers numerically solved optimal control~\citep{WangEtAl2012}, or analytical optimal control such as ECMS~\citep{RezaeiEtAl2017}, which utilizes Pontryagin’s minimum principle (PMP)~\citep{sciarretta2014}, delivers fast real-time solutions and provides global guarantees under convex formulations \citep{bouwman:2017}, and rule-based methods \citep{hofman2006} offer robust sub-optimal performance. 
However, despite their relevance, such optimal and sub-optimal supervisory strategies have not been transferred to aerospace hybrid systems, which involve fundamentally different constraints and mission profiles.

\ifjournal
The analysis and control of propulsion systems are addressed at two levels: transient dynamics~\citep{WangEtAl2017} or steady-state operation~\citep{WenEtAl2026}. The transient models capture the dynamic response of the system and they are essential for the trajectory planning~\citep{CotsEtAl2018} and feedback control~\citep{SunEtAl2021}, but they result in high-dimensional nonlinear control problems which are computationally expensive. On the other hand, in decentralized control systems, the steady state formulation focuses on the quasi-static relationship between the power requirement and the subsystems operation, neglecting the dynamic effects strategically while maintaining the dominant thermodynamic interconnections. In hybrid SOFC/GT systems, this distinction is relevant due to the slow thermal dynamics of the fuel cell compared to the faster response of the gas turbine~\citep{ZaccariaEtAl2016}. Therefore, in this work we adopt a quasi-static modeling framework based on the steady-state operation to derive the analytical power management policy. 
\fi

Although extensive work exists on power‑tracking control, degradation mitigation, and multi‑loop coordination for SOFC/GT systems, to the best of the authors' knowledge no prior studies have addressed optimal supervisory power management aimed specifically at minimizing fuel consumption in SOFC/GT hybrid propulsion for aerospace applications. Existing optimal control methods either target efficiency or power tracking rather than fuel minimization, neglect critical SOFC thermal constraints, or have only been applied in automotive or stationary contexts.

\emph{Statement of Contribution:}
In this study, we formulate an optimal static feedforward control problem designed to determine such a policy, which is subsequently solved using the KKT conditions. This paper makes three main contributions: (i) We identify computationally tractable affine and quadratic surrogate models of the propulsion system in algebraic formulation by performing a parametric sweep of a high-fidelity simulation of the system in steady-state operation. The constraints for the regression curvature ensures that these models maintain the curvature properties for the analytical optimization. (ii) By using the resulted surrogate models with quasi-static approach, we formulate the static minimum-fuel problem that incorporates the system's thermal and operational constraints. By using the monotonicity and curvature behavior of the lifting variables, the problem can be expressed in a partially convex, one dimensional form. (iii) We eventually derive the optimal power management policy analytically by using the Karush-Kuhn-Tucker optimality conditions, resulting in a closed form supervisory policy that requires no iterative computations. This facilitates the real-time implementation while preserving consistency with physical limits.

\ifjournal
A preliminary version of this study 
is submitted to the 2026 Congress of the International Council of the Aeronautical Sciences~\citep{pak2026}. In the present extended version, we provided a comprehensive review of the relevant literature and offer a detailed exposition of the model structure for different phases. Based on the modeling assumptions and the structure inferred from the data we establish an argument regarding convexity of the resulting optimal problem. The purpose is to guarantee optimality in a specific domain. Furthermore, a complete derivation of the proposed policy is included in the Appendix to support the claims made in the main text.  
\fi

\emph{Organization:} The study is structured as follows: Section 2 outlines the problem and system models. Section 3 describes the optimal control formulation and analytical solution, Section 4 presents the results briefly to validate the methodology, Section 5 suggests future approaches, and Section 6 concludes with key findings and future directions.

\section{Problem Statement} \label{sec:ps}
\FloatBarrier
This section introduces the propulsion system architecture. Then, we present the system mathematical modeling approach, and then we formulate the objective function and the optimization problem for the optimal control problem.
\subsection{System Description} \label{sub:sd}
The hybrid propulsion system illustrated in~\cref{fig:scheme} consists of a GT and an electrical subsystem. The electrical subsystem is composed of an SOFC, an inverter, and an electrical motor.
The GT features low- and high-pressure compressor stages (LPC, HPC). A portion of the compressed air is routed to the SOFC cathode after warm-up through a heat exchanger (HEX), while hydrogen is supplied to the anode after a HEX. The electrochemical reactions in the SOFC generate electricity and produce steam. The SOFC off-gases are recirculated to the GT combustion chamber after moving through the HEX, where hydrogen is burned.
The recirculation of the air through the propulsion system introduces strong thermodynamic coupling between the GT and the SOFC. The supervisory power management policy determines the power distribution law between the gas turbine and the fuel cell.

\ifjournal
The following assumption had been taken into account for this problem:
	(i) all flight phases are modeled as steady state operating points with no dynamic state evolution, (ii) the surrogate models are assumed valid within the admissible power and thermal operating envelopes defined in~\cref{eq:optp}, (iii) the admissible fuel cell power range lies strictly below the peak of the concave SOFC inlet temperature surrogate.
\fi
\subsection{Mathematical Modeling}\label{sub:mod}
\ifjournal
	For modeling we first generate steady-state simulation data through a parametric sweep of a high fidelity thermodynamic model. Then, we identifying dominant dependencies to power allocation between gas turbine and the SOFC,
		 regression of affine or quadratic surrogate models, 
		we enforce the curvature constraints to preserve the geometry properties for the optimal control problem.

\fi
Based on the steady state dataset, we approximate the key thermodynamic and flow variables, i.e., lifting and input variables, of the hybrid SOFC/GT system using a quasi-static mapping from the gas turbine output power $P_{\mathrm{gt}}$ and the SOFC electrical power $P_{\mathrm{fc}}$. For each flight phase, we select the simplest functional form (affine or quadratic) that provides sufficient accuracy while maintaining the convexity and concavity structures of each variable for the subsequent optimization problem. This approach acknowledges that the operating conditions and the dominant system nonlinearities vary substantially between take off, top of climb, and cruise. 
We interpret each mapping as a reduced-order surrogate of the high-fidelity model, obtained by projecting the steady-state response manifold onto the low dimensional space of $P_{\mathrm{gt}}, P_{\mathrm{fc}}$. By defining $\mathbf{P}_k = \begin{bmatrix} P_{\mathrm{gt},k} & P_{\mathrm{fc},k}\end{bmatrix}^{\top}$, affine coefficients are obtained using
\begin{equation}\label{eq:lp}
\min_{\{\theta_1,\theta_2,\theta_0\}} 
\sum_{k=1}^{N} 
\left\| y_k - 
\left(\begin{bmatrix} \theta_1 & \theta_2 \end{bmatrix}
\mathbf{P}_k
+ \theta_0\right)
\right\|_2, \\
\end{equation}
while quadratic terms subject to curvature constraints are estimated via semidefinite programming of the form shown in~\cref{eq:qp}
\begin{equation}\label{eq:qp}
	\begin{aligned}
		&\min_{\{Q,\mathbf{q},q_0\}} 
		\sum_{k=1}^{N} 
		\left\| y_k - 
		\left(\mathbf{P}_k^\top
		Q
		\mathbf{P}_k
		+ \mathbf{q}^\top 
		\mathbf{P}_k
		+ q_0\right)
		\right\|_2,\\
		& \ \ \ \ \ \ \ \ \ \ \ \textnormal{s.t. }\  Q \succeq 0 \ \textnormal{or } \ Q \preceq 0.
	\end{aligned}
\end{equation}

The semidefinite constraints on $Q$ enforce convexity or concavity of the quadratic maps, ensuring compatibility with the optimality of the power management problem formulation.

Imposing convex or concave curvature on the surrogate models preserves the structure required for the partial convex optimization formulations in Section~\ref{sec:met}, guaranteeing global optimality and efficient solution of the power split optimization problem.
The model output vector is
\begin{equation}
\mathbf{x} \coloneqq
\begin{bmatrix}
\dot{m}_{\mathrm{b}} &
\dot{m}_{\mathrm{f,fc}} &
\dot{m}_{\mathrm{f,gt}} &
T_{\mathrm{in}} &
T_{\mathrm{hpc}} &
T_{\mathrm{et}} &
T_{\mathrm{out}}
\end{bmatrix}^{\!\top},
\end{equation}
where $T_{\mathrm{in}}$ is the SOFC inlet temperature, $T_{\mathrm{hpc}}$ is the HPC temperature, $T_{\mathrm{et}}$ is the turbine outlet temperature and $T_{\mathrm{out}}$ is the SOFC outlet temperature. Each variable $x_j$ of $\mathbf{x}$ is parameterized either in an affine manner
\begin{equation} 
x_j(\mathbf{p}) \;=\; \theta_{j,0} + \boldsymbol{\theta}_j^{\top}\,\mathbf{p},
\label{eq:affine}
\end{equation}
or in a quadratic form
\begin{equation}
x_j(\mathbf{p}) \;=\; \mathbf{p}^{\top} Q_j \mathbf{p} + \mathbf{q}_j^{\top}\mathbf{p} + q_{j,0},
\qquad Q_j \in \mathbb{S}^2,
\label{eq:quadratic}
\end{equation}
where $\mathbb{S}^2$ denotes the set of $2\times 2$ real symmetric matrices. $x_j$ is convex if $Q_j \succeq 0$ and concave if $Q_j \preceq 0$ and the vector $\mathbf{q}_j= \begin{bmatrix}q_{j,1}& q_{j,2}\end{bmatrix}^{\top}$ and scalar $q_{j,0}$ are fitting parameters. 
The selection is based on the local nonlinear characteristics of the steady state data for each flight phase.
\ifjournal
The constraints are not to perfectly model all the nonlinear thermodynamic effects, but rather to structure an approximation which creates a balance between model fidelity and optimization tractability. 
\fi
\ifjournal
\subsubsection{Take off}
During the take off phase, the fuel flow rate for the SOFC, the HPC temperature, the turbine exit temperature, and the SOFC outlet temperature are well approximated using affine functions, whereas the gas turbine fuel flow rate, the bleed air mass flow rate, and the SOFC inlet temperature exhibit nonlinear behavior and require quadratic forms. While the the quadratic models for GT fuel flow rate and the bleed air fuel flow  have convex formulations, the SOFC inlet temperature model is inherently concave, as indicated by negative semidefinite $Q$ matrix. The quadratic terms significantly improve modeling accuracy for these variables, especially under the take off high load condition. Therefore, the models for take off are
\begin{align}
\dot{m}_{\mathrm{f,fc}} = &b_0 + b_1\cdot P_{\mathrm{gt}} + b_2\cdot P_{\mathrm{fc}}, \label{eq:modelTO}\\
T_{\mathrm{et}} = &d_0 + d_1 \cdot P_{\mathrm{gt}} + d_2\cdot P_{\mathrm{fc}},\label{eq:etTO}\\
T_{\mathrm{out}} = &h_0 + h_1 \cdot P_{\mathrm{gt}} + h_2\cdot P_{\mathrm{fc}},\label{eq:toutTO}\\
\dot{m}_{\mathrm{f,gt}} = &\begin{bmatrix}
P_{\mathrm{gt}} & P_{\mathrm{fc}}
\end{bmatrix} Q_{\mathrm{f}} \begin{bmatrix}
P_{\mathrm{gt}} \\ P_{\mathrm{fc}}
\end{bmatrix} + \mathbf{q_{\mathrm{f}}}^{\top} \begin{bmatrix}
P_{\mathrm{gt}} \\ P_{\mathrm{fc}}
\end{bmatrix} + q_{\mathrm{f},0},\\
\dot{m}_{\mathrm{b}} = &\begin{bmatrix}
P_{\mathrm{gt}} & P_{\mathrm{fc}}
\end{bmatrix} Q_{\mathrm{a}} \begin{bmatrix}
P_{\mathrm{gt}} \\ P_{\mathrm{fc}}
\end{bmatrix} + \mathbf{q_{\mathrm{a}}}^{\top} \begin{bmatrix}
P_{\mathrm{gt}} \\ P_{\mathrm{fc}}
\end{bmatrix} + q_{\mathrm{a},0}, \label{eq:mbTO}\\
T_{\mathrm{hpc}} = &\begin{bmatrix}
	P_{\mathrm{gt}} & P_{\mathrm{fc}}
\end{bmatrix} Q_{\mathrm{h}} \begin{bmatrix}
	P_{\mathrm{gt}} \\ P_{\mathrm{fc}}
\end{bmatrix} + \mathbf{q}_{\mathrm{h}}^{\top} \begin{bmatrix}
	P_{\mathrm{gt}} \\ P_{\mathrm{fc}}
\end{bmatrix} + q_{\mathrm{h,0}},\label{eq:ThpcTO}\\
T_{\mathrm{in}} = &\begin{bmatrix}
P_{\mathrm{gt}} & P_{\mathrm{fc}}
\end{bmatrix} Q \begin{bmatrix}
P_{\mathrm{gt}} \\ P_{\mathrm{fc}}
\end{bmatrix} + \mathbf{q}^{\top} \begin{bmatrix}
P_{\mathrm{gt}} \\ P_{\mathrm{fc}}
\end{bmatrix} + q_0.\label{eq:TinTO}
\end{align}
Here, $\dot{m}_{\mathrm{f,fc}}$ is the SOFC fuel flow rate, $T_{\mathrm{hpc}}$ represents the high pressure compressor temperature, $T_{\mathrm{et}}$ corresponds to the turbine temperature, $T_{\mathrm{out}}$ refers to the SOFC outlet temperature at both anode and cathode side, $\dot{m}_{\mathrm{f,gt}}$ denotes the gas turbine fuel flow rate, $\dot{m}_{\mathrm{b}}$ is the bleed air flow rate and $T_{\mathrm{in}}$ is the SOFC inlet temperature. The coefficients $b_0,\, b_1,\, b_2$ are the regression parameters used in the approximation of $\dot{m}_{\mathrm{f,fc}}$, while $Q_{\mathrm{h}},\, \mathbf{q}_{\mathrm{h}},\, q_{\mathrm{h,0}}$ are the corresponding parameters for $T_{\mathrm{hpc}}$. Similarly, $d_0,\, d_1,\, d_2$ are the fitted coefficients for $T_{\mathrm{et}}$, and $h_0,\, h_1,\, h_2$ parameterize the model for $T_{\mathrm{out}}$. 
The~\crefrange{eq:modelTO}{eq:TinTO} model the relationship between the gas turbine output and fuel flow rate. This mapping is obtained from the high fidelity models parametric sweep and approximated using a set of convex and one monotonic concave functions.  
\subsubsection{Top of climb and cruise}
Under top of climb and cruise conditions, the system operates under more steadily and less demanding conditions. As a result, most variables, including fuel flow rates and temperatures, are approximated by affine models. Only the bleed air mass flow rate and the SOFC inlet temperature retain a nonlinear quadratic dependence. The models for top of climb and cruise are similar to take off for $\dot{m}_{\mathrm{f,fc}}, \ T_{\mathrm{et}}, \ T_{\mathrm{out}}, \ \dot{m}_{\mathrm{b}} \ $ and $T_{\mathrm{in}}$ as \cref{eq:modelTO,eq:etTO,eq:toutTO,eq:mbTO,eq:TinTO}, respectively, with different values for their coefficients. We model the rest of the variables linearly as
\begin{align} \label{eq:modelToCCR}
\dot{m}_{\mathrm{f,gt}} = &e_0 + e_1\cdot P_{\mathrm{gt}} + e_2\cdot P_{\mathrm{fc}},\\
T_{\mathrm{hpc}} = &t_0 + t_1 \cdot P_{\mathrm{gt}} + t_2\cdot P_{\mathrm{fc}}. \label{eq:ThpcTOC}
\end{align}
As in the take off case, the quadratic coefficient matrix $Q_{\mathrm{a}}$ associated with the bleed air flow rate is positive semidefinite, while the SOFC inlet temperature model again exhibits concavity, reflected by a negative semidefinite matrix $Q$. 
In the models we assumed: (i) the inlet temperature for both anode and cathode are identical, (ii) the outlet temperature for both anode and cathode are identical, (iii)  within each phase, the system performs in steady-state operation. Hence, a quasi-static mapping is appropriate.

Model fidelity is assessed using the normalized root-mean-square error (NRMSE) computed over the steady state dataset for each flight phase.
\ifjournal
Table~\ref{tab:NRME} summarizes the NRMSE values for all variables across take off, top of climb ans cruise.
All linear and semi definite programs used for surrogate model identification were solved using \textsc{MOSEK}~\citep{mosek} optimization solver. 
The models are identified using steady-state data and to ensure consistency between the regression models and the data, we explicitly define the admissible modeling domain
of the sampled operating points, means that
we ensure that all the interpolated operating points remain within the region spanned by the available data and extrapolation is not admissible.
\begin{table}[!b]
	\centering
	\caption{NRMSE for Input, Decision and Lifting Variables}
	\label{tab:NRME}
	\begin{tabular}{l|ccc}
		\toprule
		\textbf{Variable} & \textbf{Take off} & \textbf{Top of climb} & \textbf{Cruise} \\
		\midrule
		$\dot{m}_{\mathrm{f,gt}}$ & $0.6\%$ & $0.8\%$ & $0.9\%$ \\
		$\dot{m}_{\mathrm{f,fc}}$ & $0.6\%$ & $0.2\%$ & $0.4\%$ \\
		$\dot{m}_{\mathrm{b}}$    & $0.6\%$ & $1.1\%$ & $1.1\%$ \\
		$T_{\mathrm{in}}$         & $0.06\%$ & $0.2\%$ & $0.2\%$ \\
		$T_{\mathrm{hpc}}$        & $0.3\%$ & $0.3\%$ & $0.2\%$ \\
		$T_{\mathrm{et}}$         & $0.8\%$ & $0.4\%$ & $0.3\%$ \\
		$T_{\mathrm{out}}$	& $0.04\%$ & $0.5\%$ & $0.4\%$ \\
		\bottomrule
	\end{tabular}
\end{table}
\begin{figure*}[H]
	\begin{subfigure}[t]{0.48\textwidth}
		\centering
		\includegraphics[width=\linewidth]{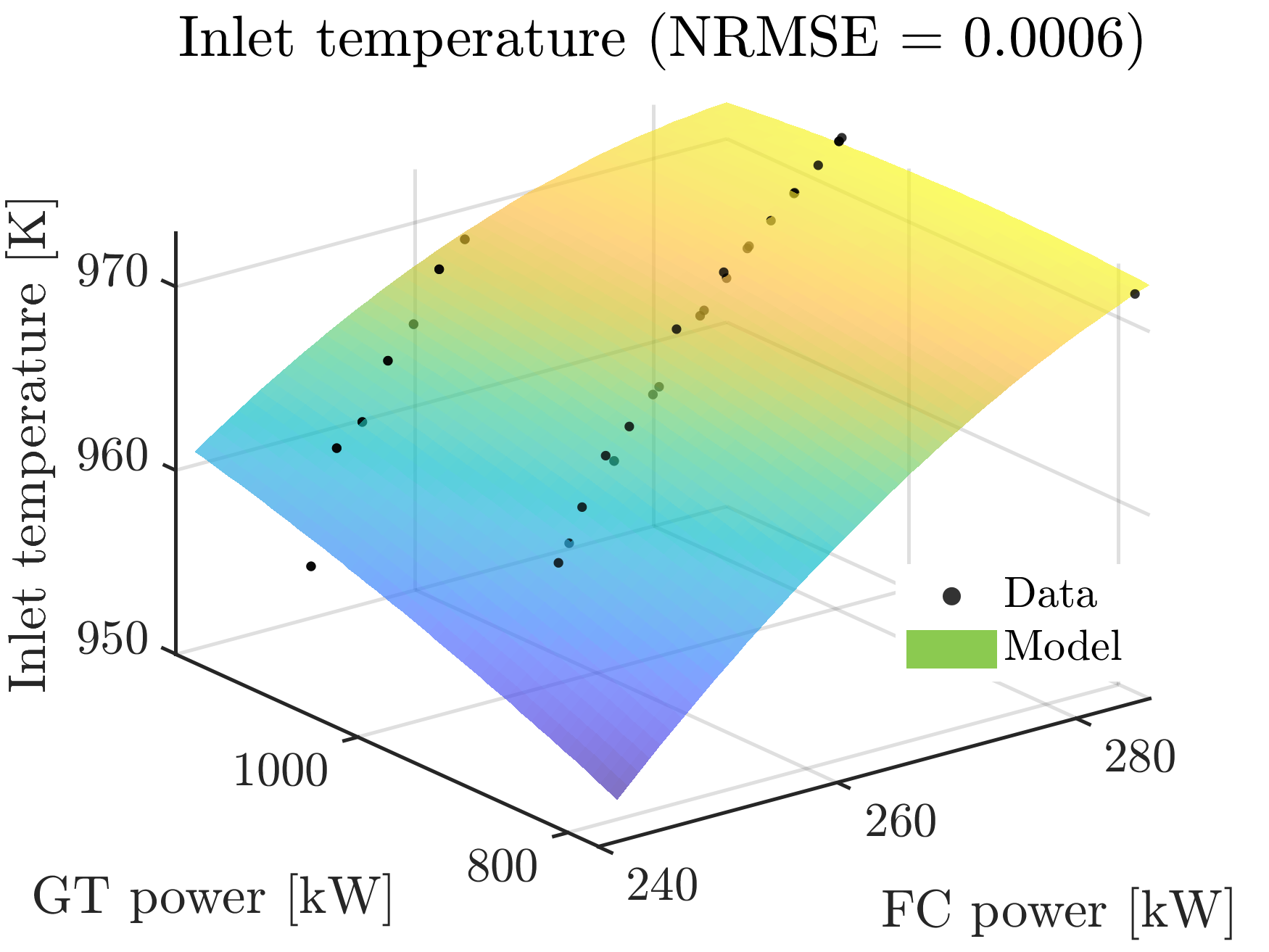}
		\caption{Take off map}
		\label{fig:ps_to}
	\end{subfigure} 
	\hfill
		\begin{subfigure}[t]{0.48\textwidth}
		\centering
		\includegraphics[width=\textwidth]{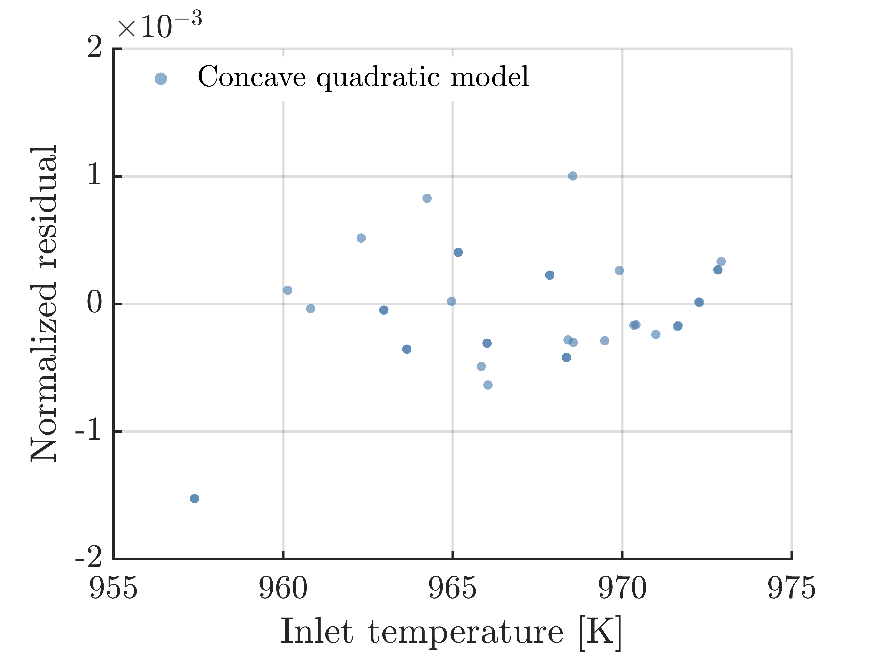}
		\caption{Take off modeling residual error}
		\label{fig:tintoer}
	\end{subfigure}
	\begin{subfigure}[t]{0.48\textwidth}
		\centering
		\includegraphics[width=\linewidth]{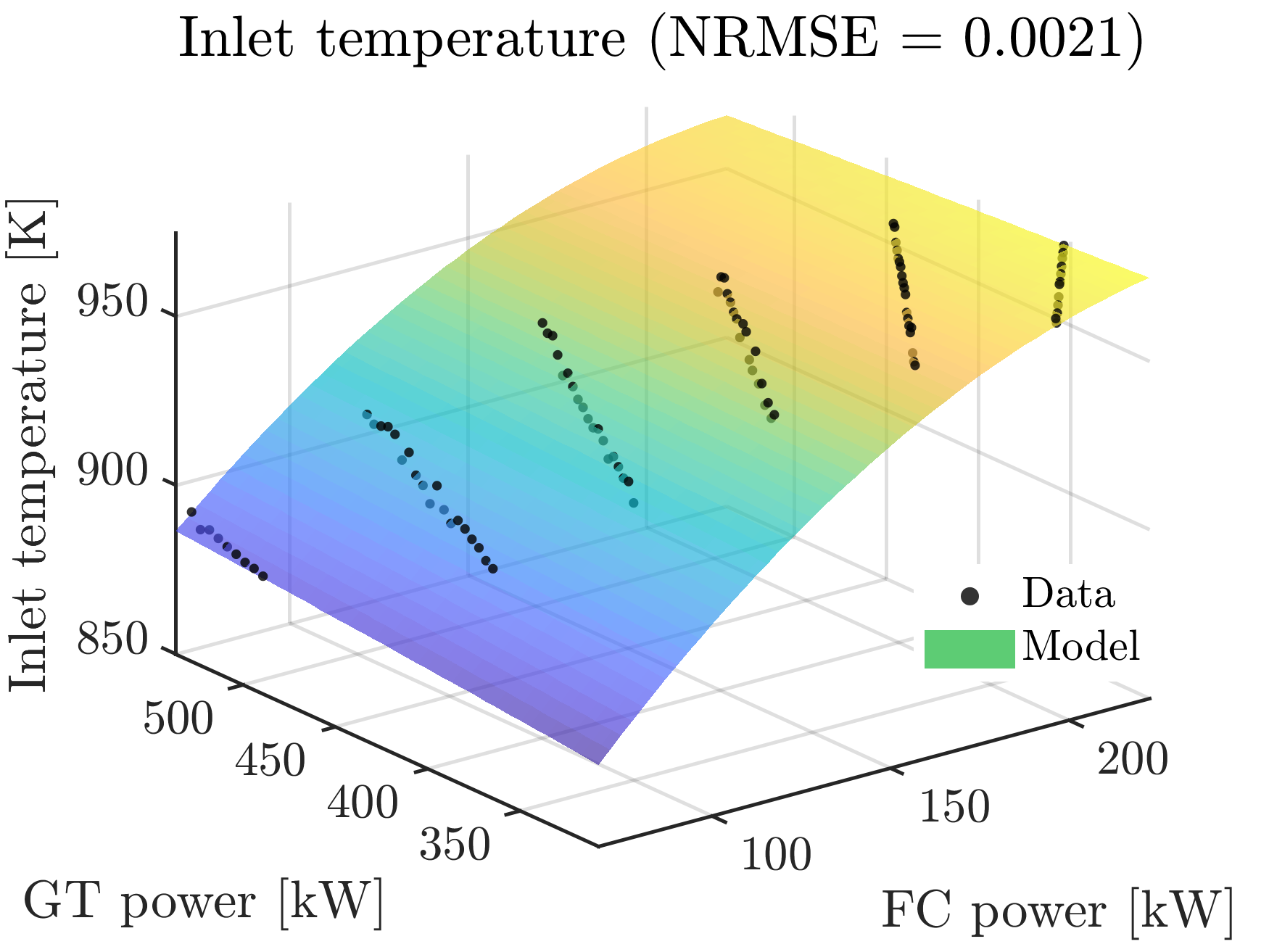}
		\caption{Top of climb map}
		\label{fig:ps_toc}
	\end{subfigure}
	\hfill
	\begin{subfigure}[t]{0.48\textwidth}
		\centering
		\includegraphics[width=\textwidth]{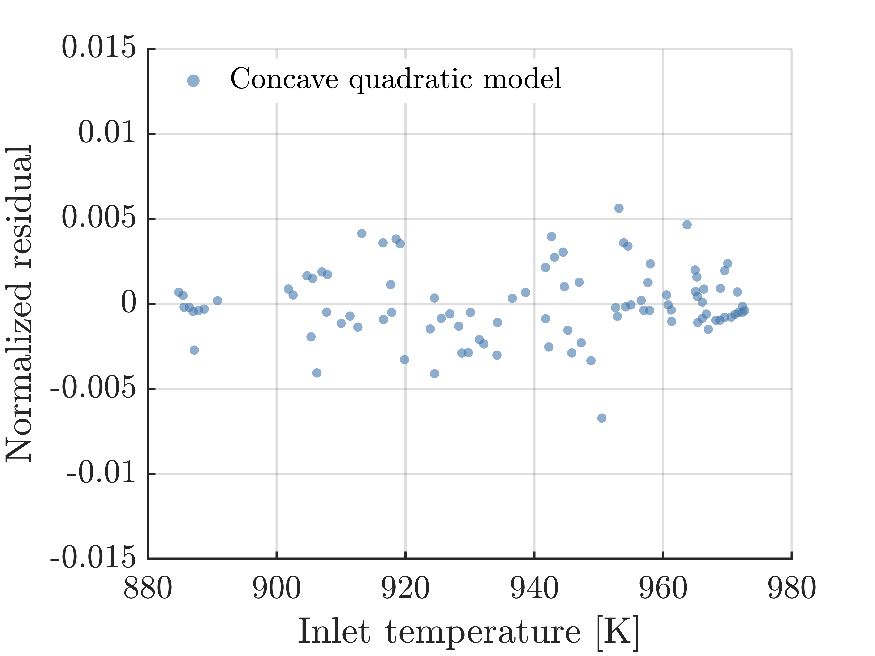}
		\caption{Top of climb modeling residual error}
		\label{fig:tintocer}
	\end{subfigure}
	\hfill
	\begin{subfigure}[t]{0.48\textwidth}
		\includegraphics[width=\linewidth]{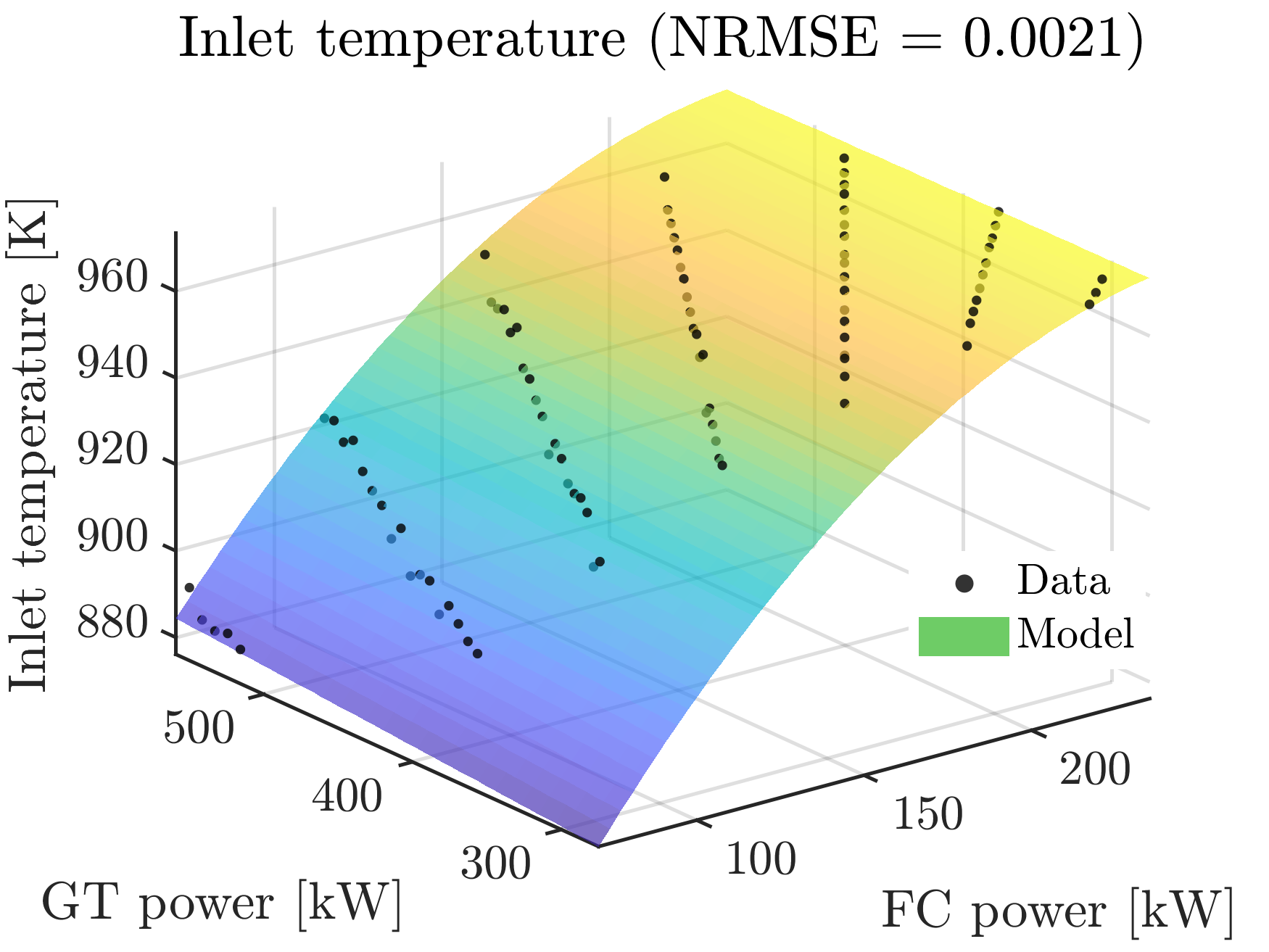}
		\caption{Cruise map}
		\label{fig:ps_cr}
	\end{subfigure}
		\hfill
	\begin{subfigure}[t]{0.48\textwidth}
		\centering
		\includegraphics[width=\textwidth]{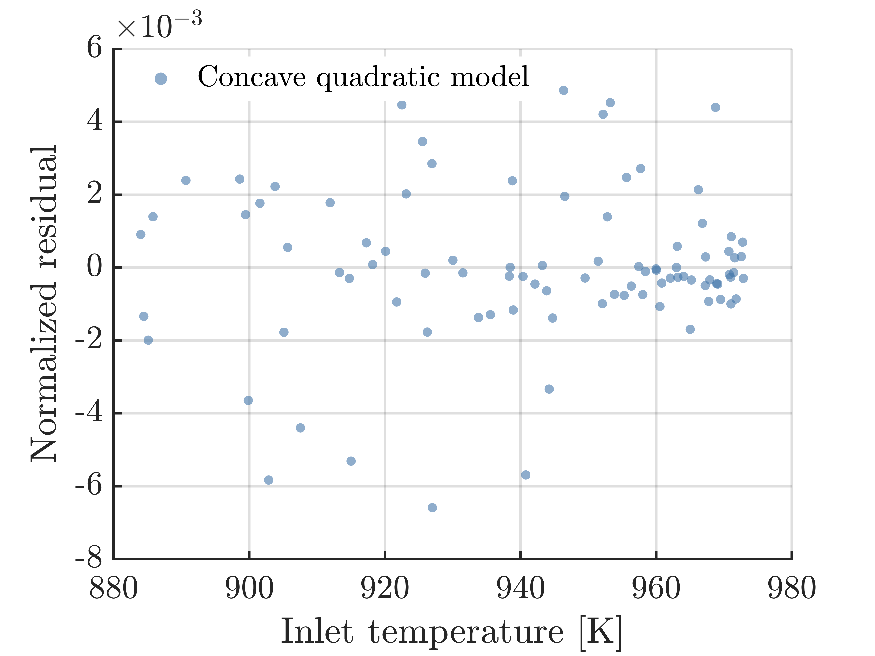}
		\caption{Cruise modeling residual error}
		\label{fig:tincrer}
	\end{subfigure}
	
	\caption{The SOFC inlet temperature map with respect to power allocation and their residual error scatter plots across flight phases.}
	\label{fig:inletT}
\end{figure*}
		%
\ifjournal
\clearpage
\begin{figure*}[H]
	\begin{subfigure}[t]{0.48\textwidth}
		\centering
		\includegraphics[width=\textwidth]{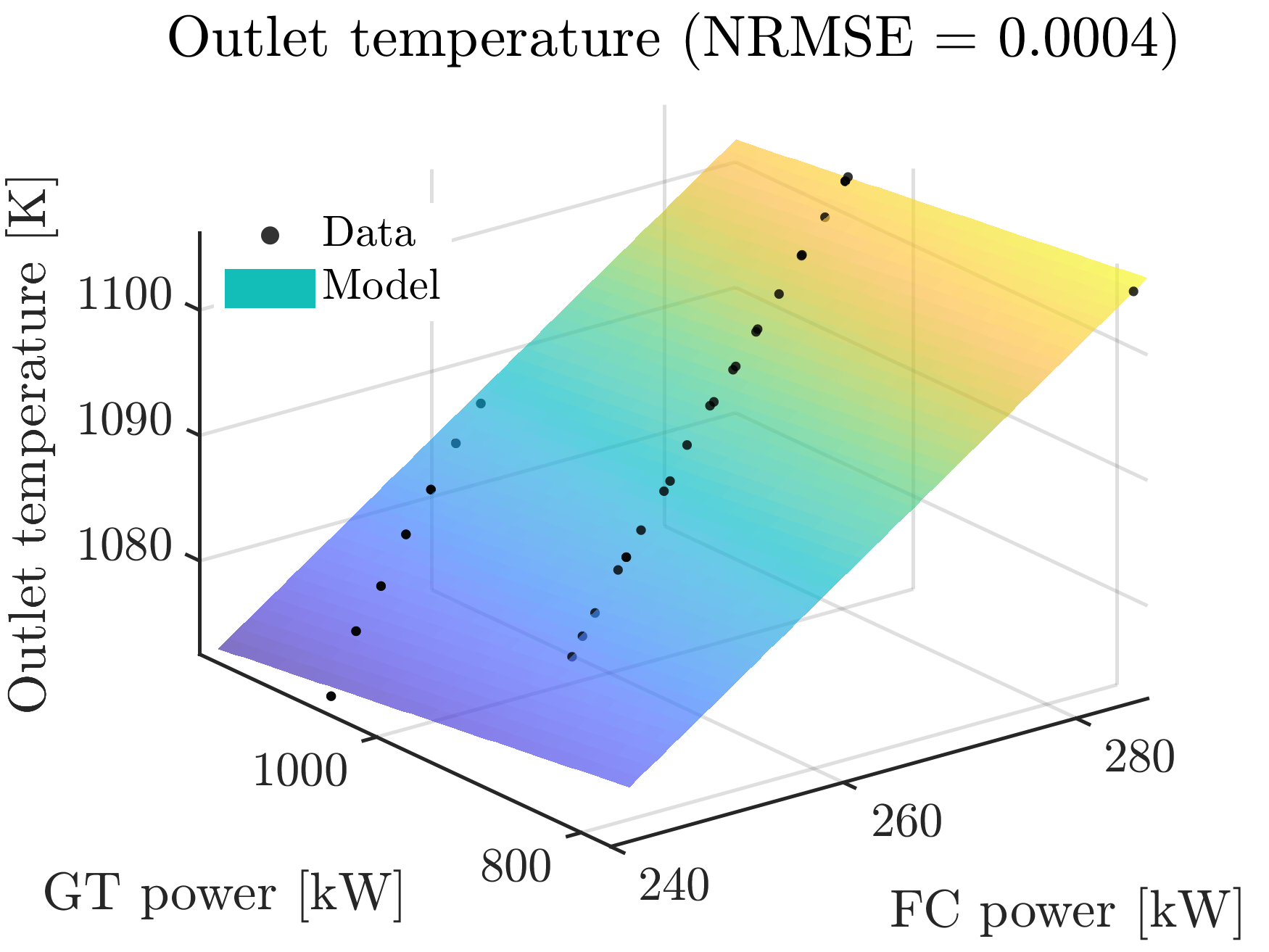}
		\caption{Take off}
		\label{fig:ps_to}
	\end{subfigure}
	\hfill
		\begin{subfigure}[t]{0.48\textwidth}
		\centering
		\includegraphics[width=\textwidth]{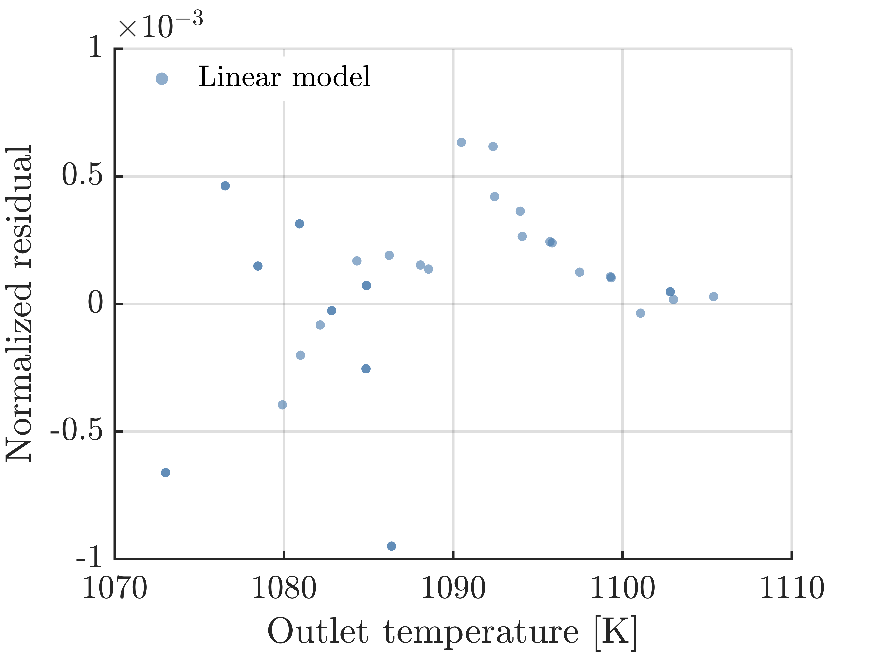}
		\caption{Take off modeling residual error}
		\label{fig:toresto}
	\end{subfigure}
	\hfill
	\begin{subfigure}[t]{0.48\textwidth}
		\centering
		\includegraphics[width=\textwidth]{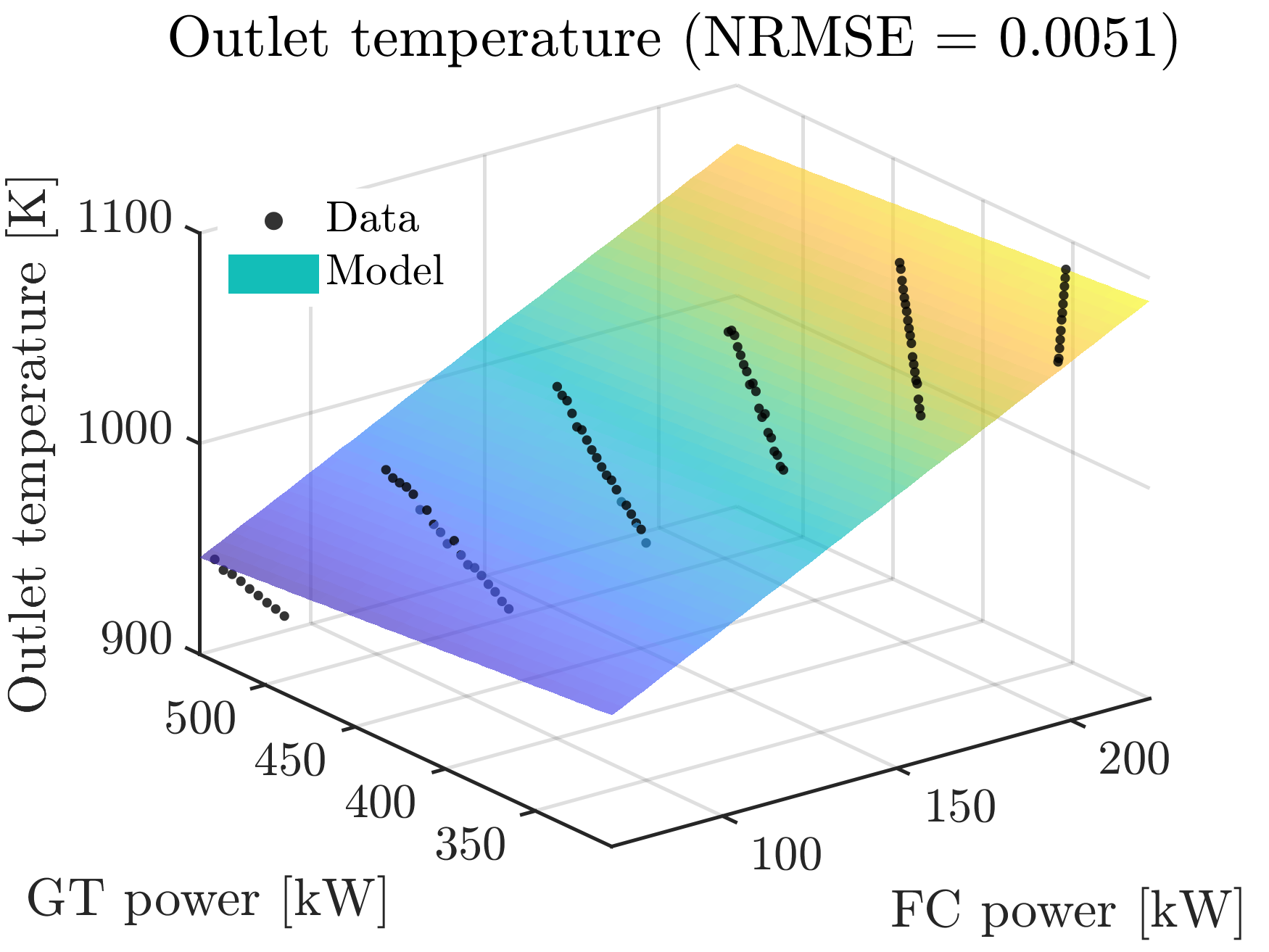}
		\caption{Top of climb}
		\label{fig:ps_toc}
	\end{subfigure}
	\hfill
		\begin{subfigure}[t]{0.48\textwidth}
		\centering
		\includegraphics[width=\textwidth]{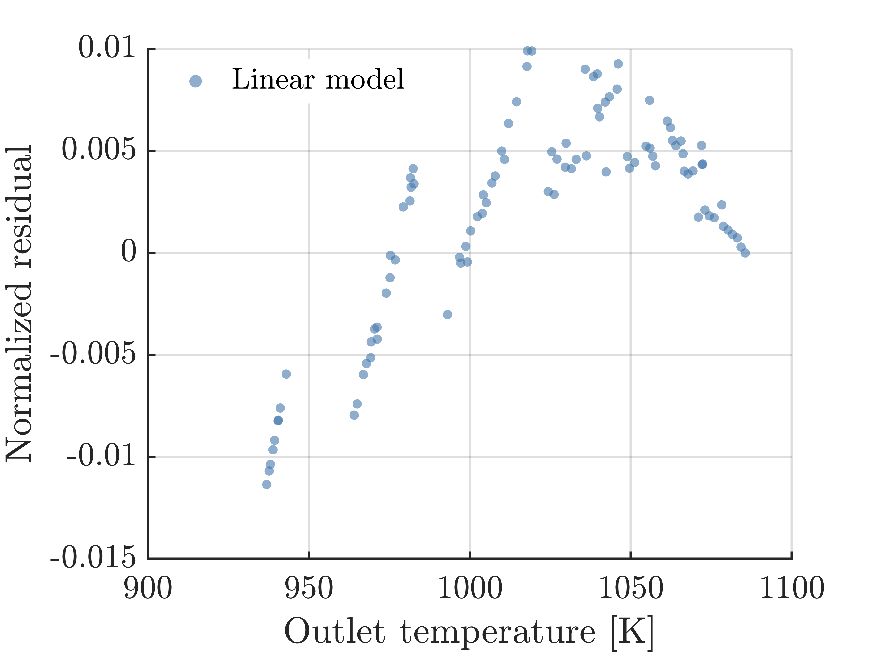}
		\caption{Top of climb modeling residual error}
		\label{fig:torestoc}
	\end{subfigure}
	\hfill
	\begin{subfigure}[t]{0.48\textwidth}
		\centering
		\includegraphics[width=\textwidth]{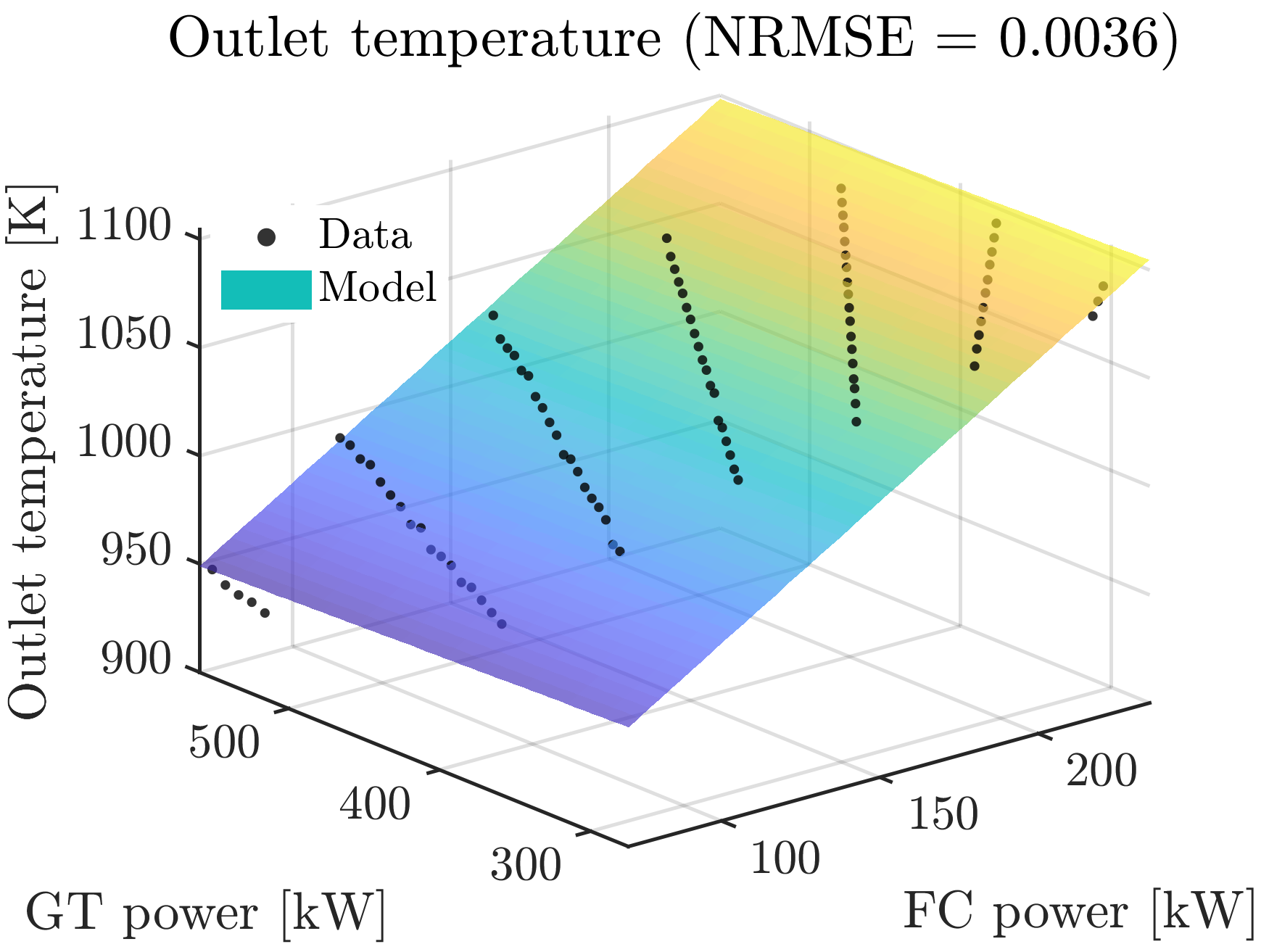}
		\caption{Cruise}
		\label{fig:ps_cr}
	\end{subfigure}
		\hfill
	\begin{subfigure}[t]{0.48\textwidth}
		\centering
		\includegraphics[width=\textwidth]{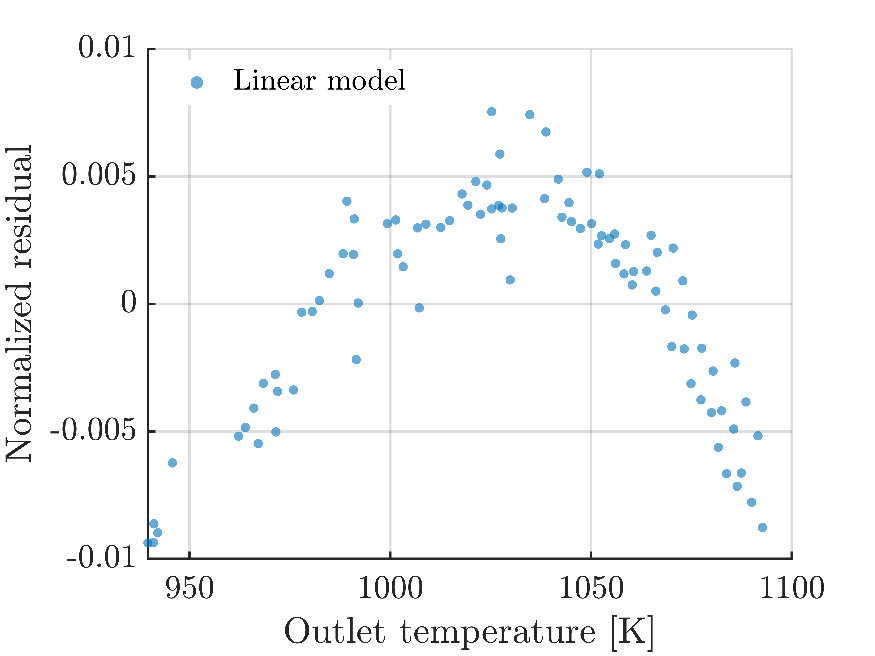}
		\caption{Cruise modeling residual error}
		\label{fig:torescr}
	\end{subfigure}
	
	\caption{SOFC outlet temperature maps as function of GT power and FC power and their residual error scatter plots across flight phases.}
	\label{fig:outletT}
\end{figure*}
		%
\fi
\clearpage
\begin{figure*}[H]
	\centering
	\begin{subfigure}[t]{0.48\textwidth}
		\centering
		\includegraphics[width=\textwidth]{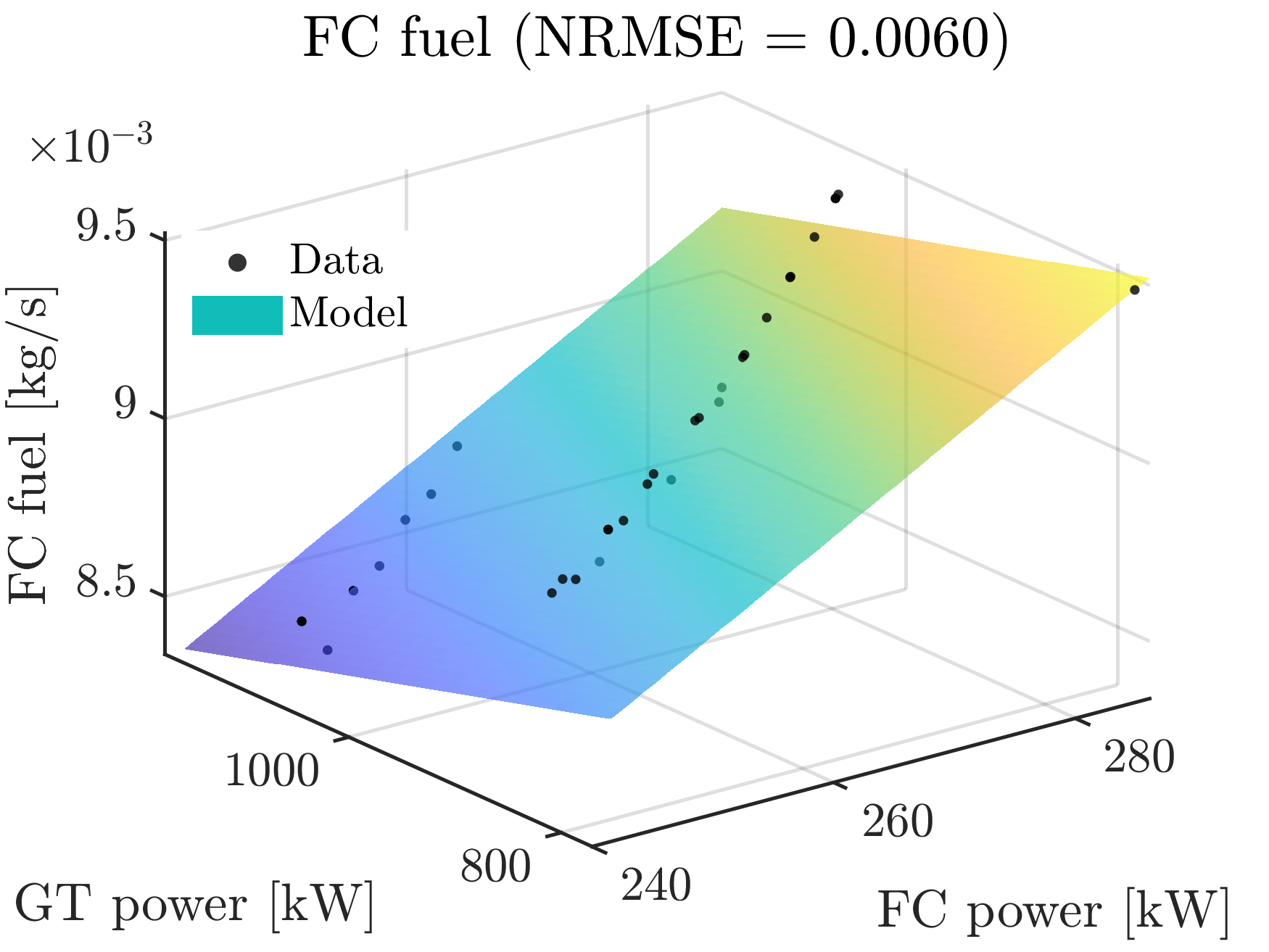}
		\caption{Take off}
		\label{fig:ps_to}
	\end{subfigure}
	\hfill
	\begin{subfigure}[t]{0.48\textwidth}
		\centering
		\includegraphics[width=\textwidth]{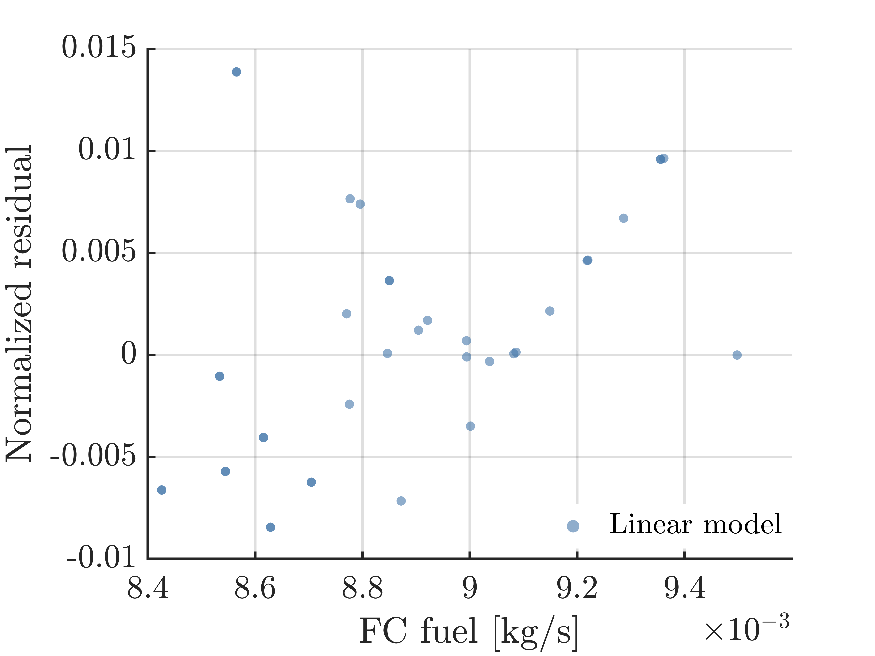}
		\caption{Take off modeling residual error}
		\label{fig:mfctoer}
	\end{subfigure}
	\hfill
	\begin{subfigure}[t]{0.48\textwidth}
		\centering
		\includegraphics[width=\textwidth]{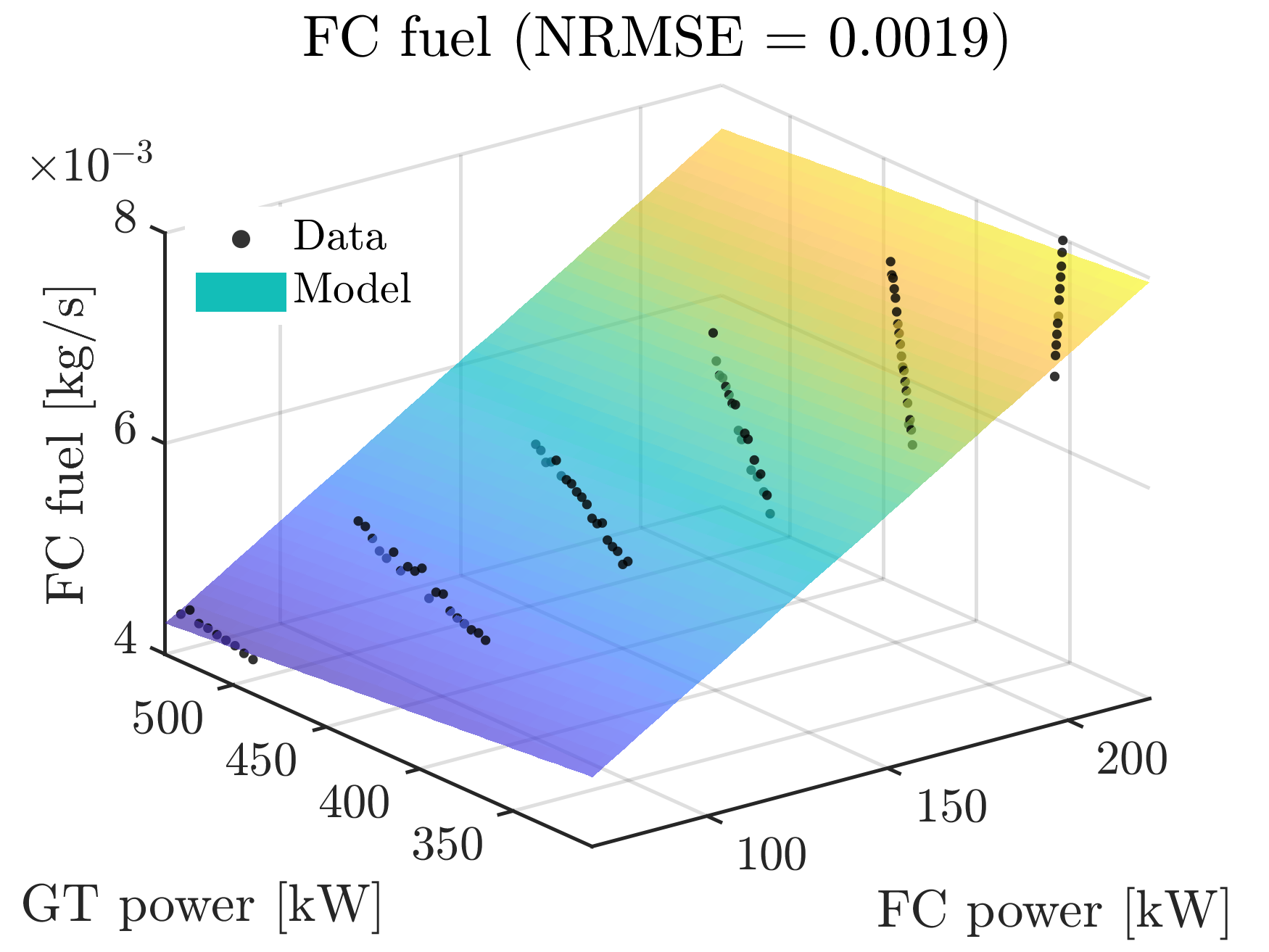}
		\caption{Top of climb}
		\label{fig:ps_toc}
	\end{subfigure}
	\hfill
	\begin{subfigure}[t]{0.48\textwidth}
		\centering
		\includegraphics[width=\textwidth]{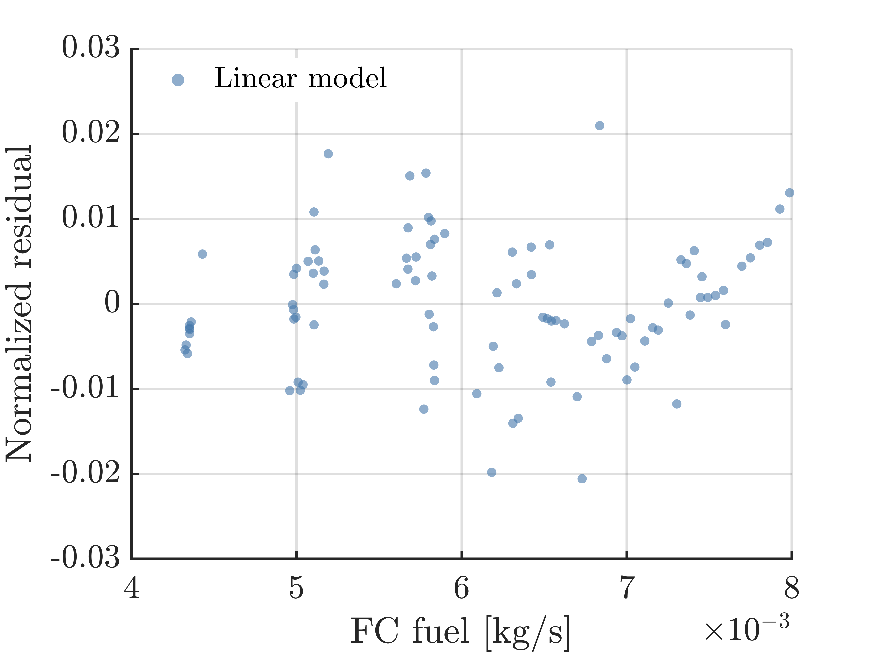}
		\caption{Top of climb modeling residual error}
		\label{fig:mfctocer}
	\end{subfigure}
	\hfill
	\begin{subfigure}[t]{0.48\textwidth}
		\centering
		\includegraphics[width=\textwidth]{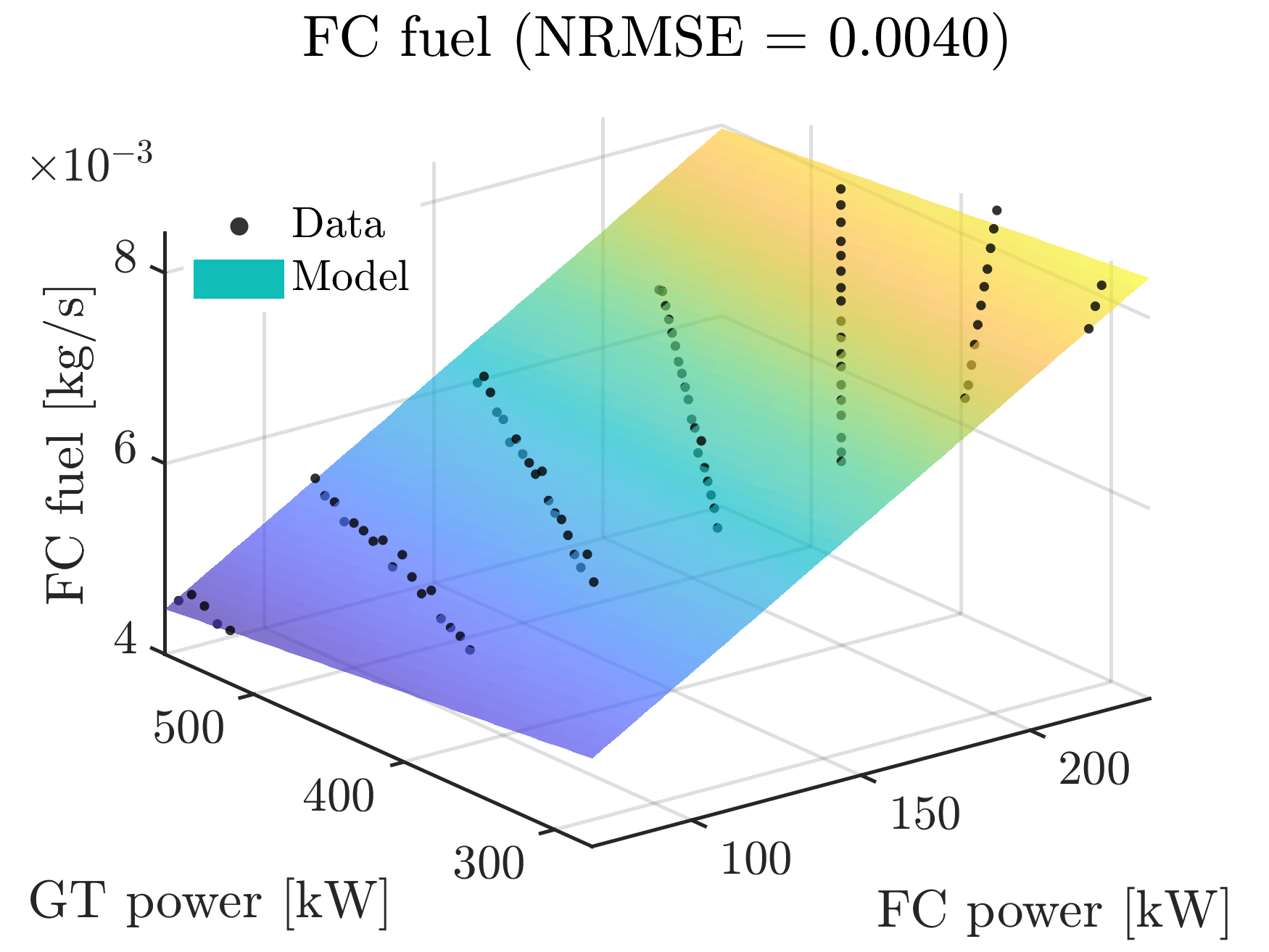}
		\caption{Cruise}
		\label{fig:ps_cr}
	\end{subfigure}
	\hfill
	\begin{subfigure}[t]{0.48\textwidth}
		\centering
		\includegraphics[width=\textwidth]{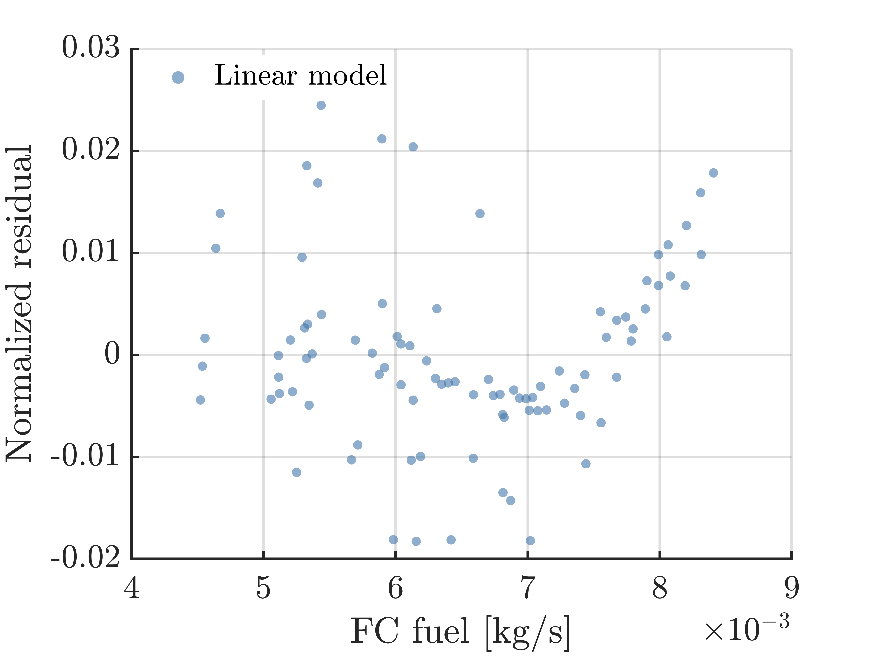}
		\caption{Cruise modeling residual error}
		\label{fig:mfccrer}
	\end{subfigure}
	
	\caption{SOFC fuel flow rate maps  as function of GT power and FC power  and their residual error scatter plots across flight phases.}
	\label{fig:fuelFC}
\end{figure*}
		%
\begin{figure*}[H]
	\centering
	\begin{subfigure}[t]{0.48\textwidth}
		\centering
		\includegraphics[width=\textwidth]{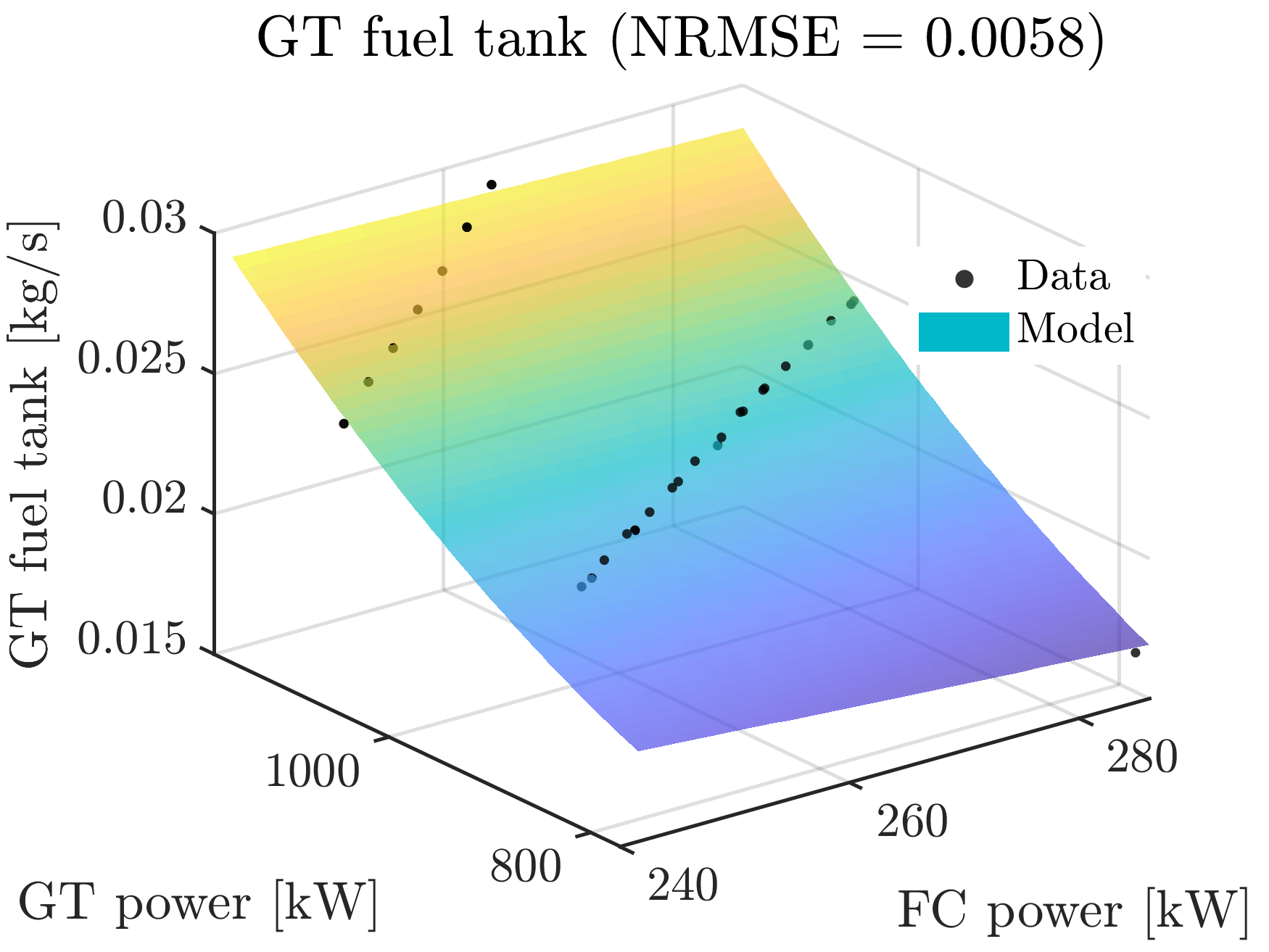}
		\caption{Take off}
		\label{fig:ps_to}
	\end{subfigure}
	\hfill
	\begin{subfigure}[t]{0.48\textwidth}
		\centering
		\includegraphics[width=\textwidth]{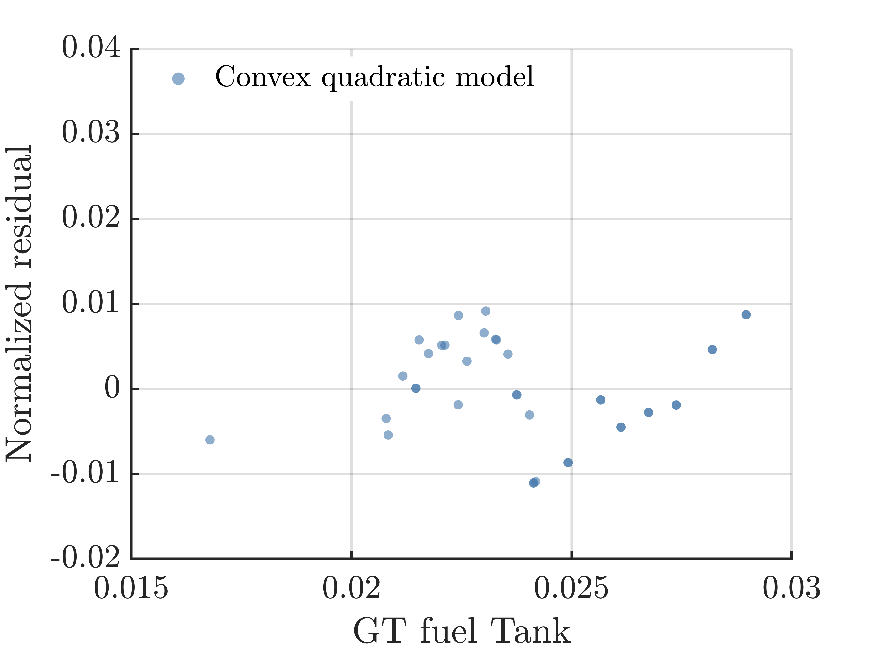}
		\caption{Take off modeling residual error}
		\label{fig:mgttoer}
	\end{subfigure}
	\hfill
	\begin{subfigure}[t]{0.48\textwidth}
		\centering
		\includegraphics[width=\textwidth]{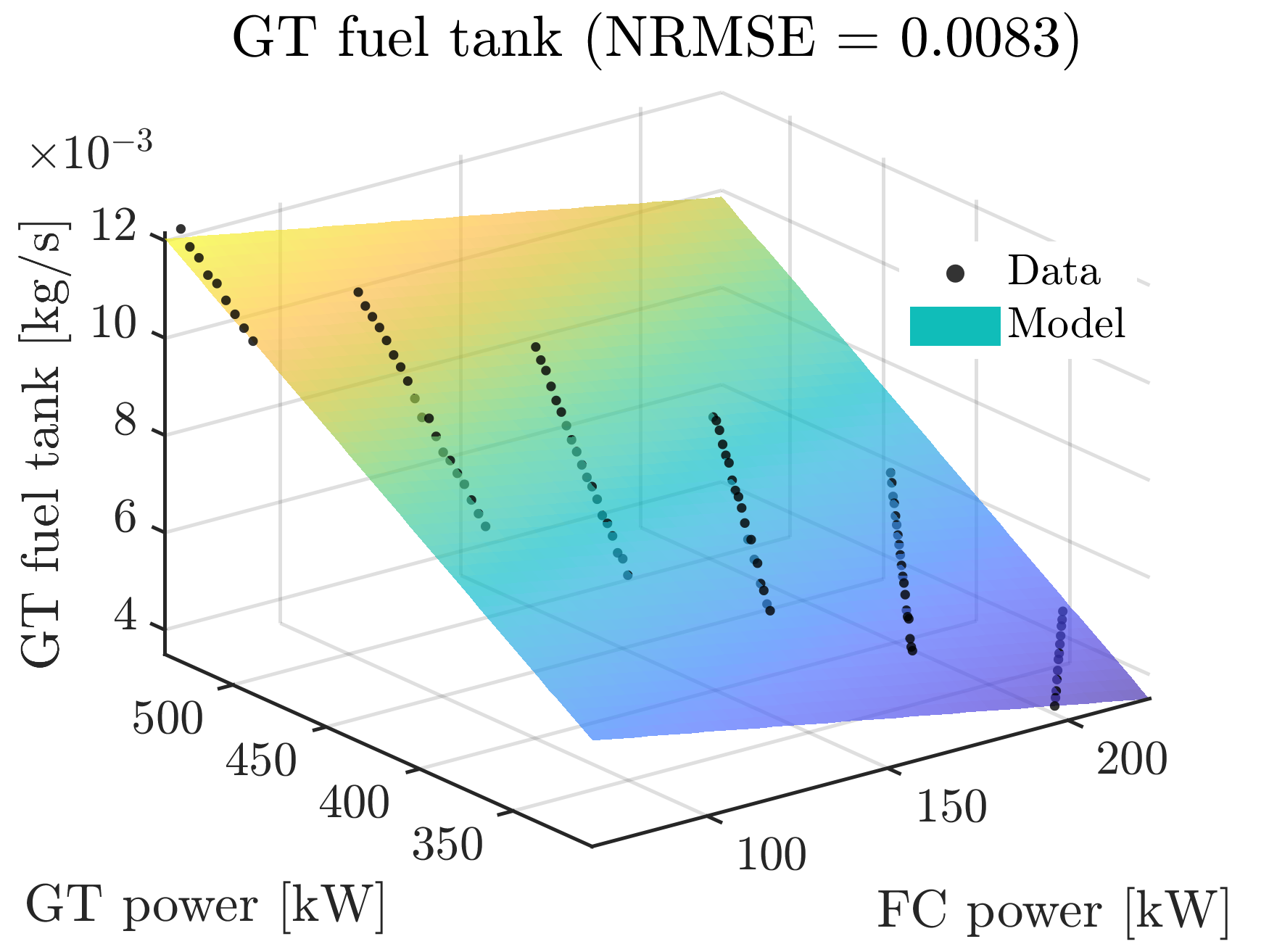}
		\caption{Top of climb}
		\label{fig:ps_toc}
	\end{subfigure}
	\hfill
	\begin{subfigure}[t]{0.48\textwidth}
		\centering
		\includegraphics[width=\textwidth]{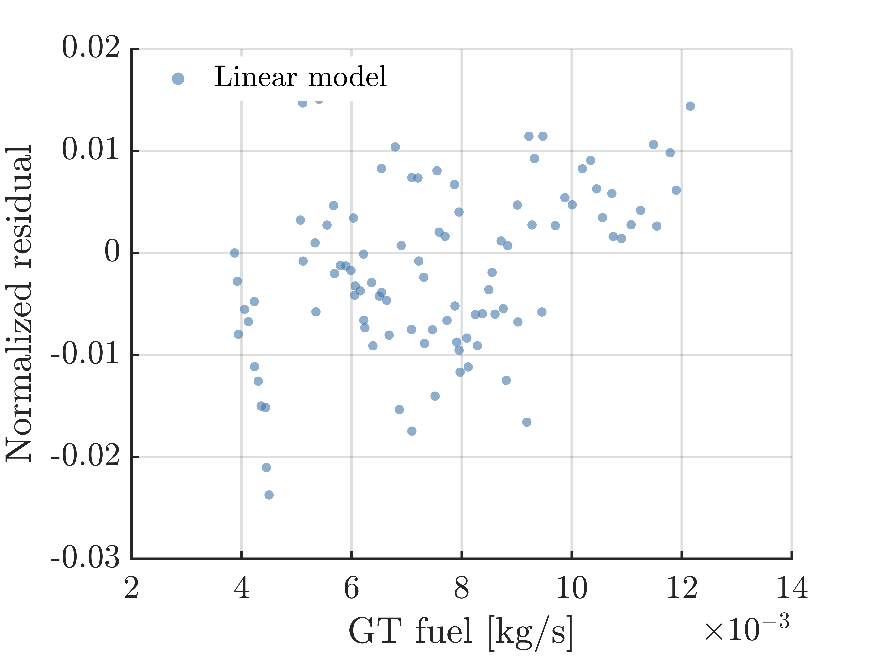}
		\caption{Top of climb modeling residual error}
		\label{fig:mgttocer}%
	\end{subfigure}
	\hfill
	\begin{subfigure}[t]{0.48\textwidth}
		\centering
		\includegraphics[width=\textwidth]{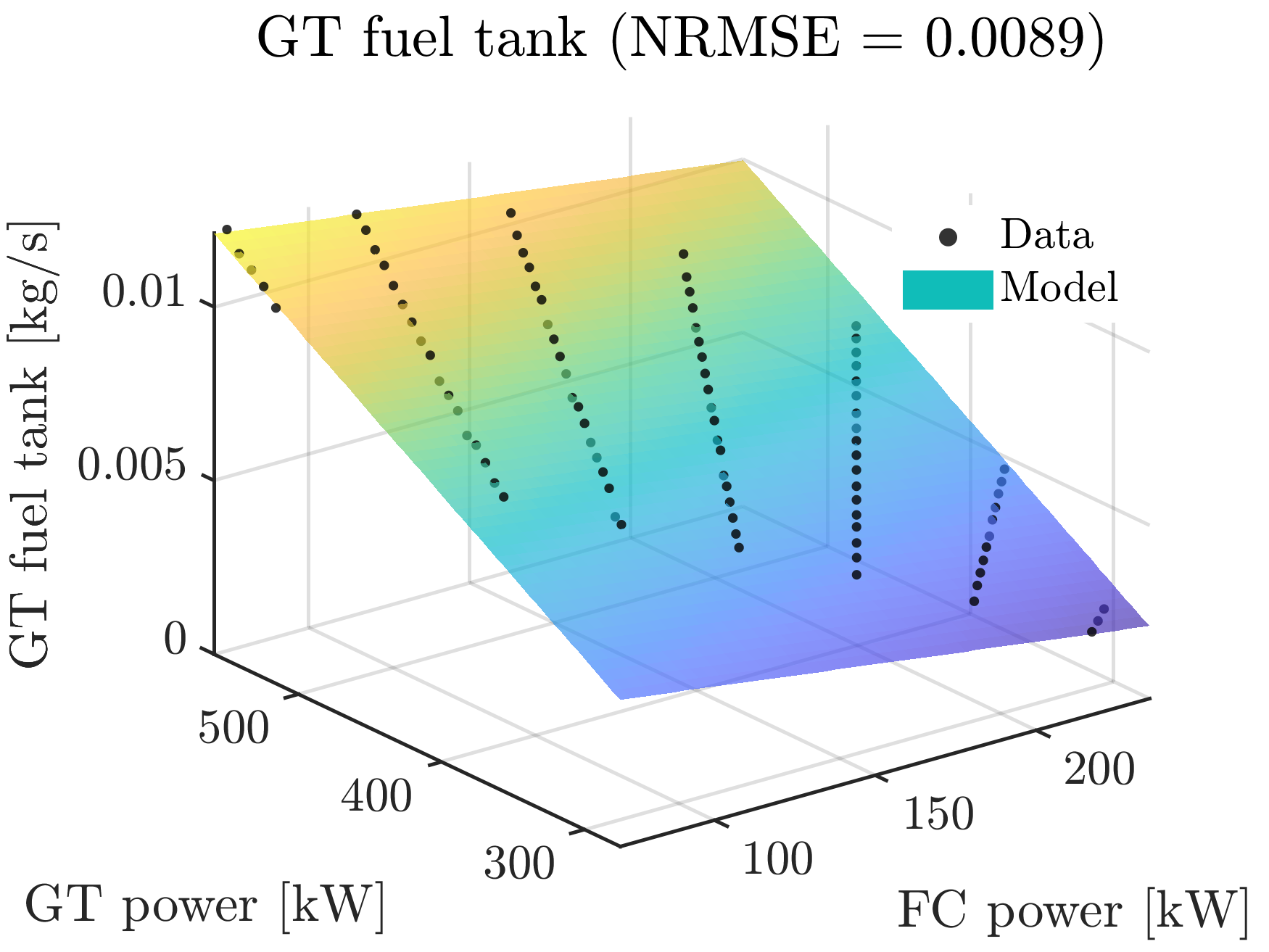}
		\caption{Cruise}
		\label{fig:ps_cr}
	\end{subfigure}
	\hfill
	\begin{subfigure}[t]{0.48\textwidth}
		\centering
		\includegraphics[width=\textwidth]{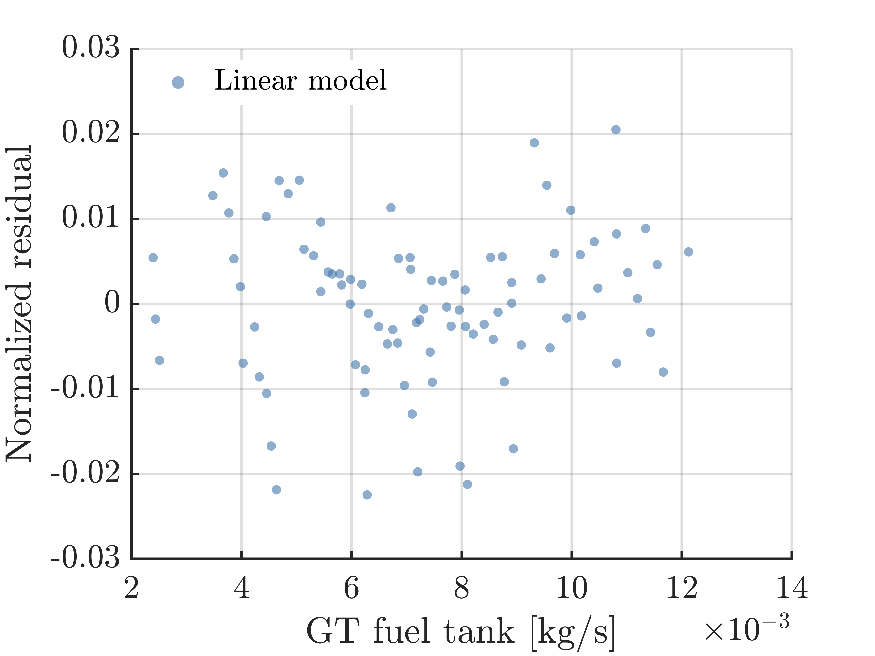}
		\caption{Cruise modeling residual error}
		\label{fig:mgtcrer}
	\end{subfigure}
	
	\caption{Gas turbine fuel flow rate maps as function of GT power and FC power and their residual error scatter plots across flight phases.}
	\label{fig:fuelGT}
\end{figure*} 
Additionally,~\crefrange{fig:inletT}{fig:fuelGT} present three dimensional surface maps illustrating how subsystem power allocation influences key outputs of the propulsion system across different flight phases. The horizontal axes present the values for gas turbine power and fuel cell power, while the vertical axis shows the resulting outputs. On the right side of the maps, the residual error of the models are shown for each variable.

We show the SOFC inlet temperature map in~\cref{fig:inletT}. The SOFC inlet temperature varies approximately between $840 \ K$ to $973.15 \ K$, and increases as both GT power and FC power rise. Although the  surface shows a noticeable upward slope with fuel cell power in all three phases, showing higher fuel cell power leads to hotter SOFC inlet condition, the temperature gradient is steepest along the GT power during take off comparing to top of climb and the cruise phases. This suggests that gas turbine performance has a slight influence on the inlet temperature of SOFC in take off.

Figure~\ref{fig:outletT} depicts the surfaces that quantify the SOFC outlet temperature changes as power is distributed between the GT and SOFC. At take off, higher total power demand causes outlet temperature to rise sharply with increasing the FC power, and slightly decreasing with rising GT power, reflecting the thermal behavior of the SOFC and the cooling impact of the air from GT. During top of climb, the temperature surface shows smoother and less steep changes, corresponding to reduced power demand.

Furthermore,~\cref{fig:fuelFC} shows the dependency of fuel cell fuel mass flow rate on the power allocation. Across all three flight phases the fuel consumption increases linearly with fuel cell power. During take off, the SOFC fuel flow rate increases significantly with increasing fuel cell power and decreases gradually with rising gas turbine power, while in top of climb and in cruise the gas turbine power does not affect the fuel flow required for the fuel cell.

Finally,~\cref{fig:fuelGT} depicts the surface map for GT fuel flow rate changing with the power distribution. The maps for cruise and top of climb shows that the fuel flow rate to gas turbine from tank rises linearly by increasing GT power and slightly decreasing by rising the fuel cell power; however, it has a convex behavior in take off, showing almost no dependency on the fuel cell power but a rapid quadratic change with rising the gas turbine power. This behavior reflects the nonlinear efficiency of the gas turbine under high load condition. 

To validate the models,~\cref{fig:tintoer,fig:tintocer,fig:tincrer,fig:toresto,fig:torestoc,fig:torescr,fig:mfctoer,fig:mfctocer,fig:mfccrer,fig:mgttoer,fig:mgttocer,fig:mgtcrer} 
 present the scatter plots of the residual errors of the surrogate models for key lifting variables across take off, top of climb and cruise conditions 
as
\begin{equation}
	r_{\mathrm{e}}=\frac{(x-\hat{x})}{x_{\mathrm{max}}},
\end{equation}
where $\hat{x}$ is the estimated value of the variable from the approximation models and $x_{\mathrm{max}}$ is the maximum value from the dataset for each variables.
The residual analysis shows that across all variables and flight phases, the residuals are centered around zero and exhibits no discernible bias, indicating that the surrogate models fo not systematically over or under predict the behavior of the system. Moreover, the absence of visible patterns in~\cref{fig:tintoer,fig:tintocer,fig:tincrer,fig:mfctoer,fig:mfctocer,fig:mfccrer,fig:mgttoer,fig:mgttocer,fig:mgtcrer} 
suggests that the order of the system and nonlinearities are sufficiently captured. However, the presence of a low order pattern in~\cref{fig:toresto,fig:torestoc,fig:torescr} 
suggests that higher order thermal model is relevant at the studied power ranges. We neglect the nonlinearity effect as the residual magnitude remains small and does not affect power management results. 
These surface maps demonstrate that the surrogate models not only achieve sufficient accuracy but also capture relevant physical trends required for the optimal controller.
\else
Table~\ref{tab:NRME} summarizes the NRMSE values for all lifting variables across take off, top of climb ans cruise.
All linear and semi definite programs used for surrogate model identification were solved using \textsc{MOSEK}~\citep{mosek}. 
\fi
\ifjournal
\else
\begin{table}[h!]
	\centering
	\caption{NRMSE for Lifting Variables}
	\label{tab:NRME}
	\begin{tabular}{l|ccc}
		\toprule
		\textbf{Variable} & \textbf{Take off} & \textbf{Top of climb} & \textbf{Cruise}\\
		\midrule
		$\dot{m}_{\mathrm{f,gt}}$ & $0.6\%$ & $0.8\%$ & $0.9\%$\\
		$\dot{m}_{\mathrm{f,fc}}$ & $0.6\%$ & $0.2\%$ & $0.4\%$\\
		$\dot{m}_{\mathrm{b}}$    & $0.6\%$ & $1.1\%$ & $1.1\%$\\
		$T_{\mathrm{in}}$         & $0.06\%$ & $0.2\%$ & $0.2\%$\\
		$T_{\mathrm{hpc}}$        & $0.3\%$ & $0.3\%$ & $0.2\%$\\
		$T_{\mathrm{et}}$         & $0.8\%$ & $0.4\%$ & $0.3\%$\\
		$T_{\mathrm{out}}$	& $0.04\%$ & $0.5\%$ & $0.4\%$ \\
		\bottomrule
	\end{tabular}
\end{table}
\fi
\ifjournal

\else
Additionally,\ifjournal~\cref{fig:inletT, fig:outletT, fig:fuelFC, fig:fuelGT} \else~\cref{fig:inletT,fig:fuelFC,fig:fuelGT} \fi present three dimensional surface maps illustrating how subsystem power allocation influences key outputs of the propulsion system across different flight phases. The horizontal axes present the values for gas turbine power and fuel cell power, while the vertical axis shows the resulting outputs. These figures exhibit important information about how each variable increases or decreases in response to the change in power distribution between GT and SOFC.
\fi

\ifjournal
\else
	
To validate the models,~\cref{fig:Tinres,fig:mfcres,fig:mfgt_res} present the scatter plots of the normalized residual errors of the surrogate models for key lifting variables across take off, top of climb and cruise conditions
 as
\begin{equation}
	r_{\mathrm{e}}=\frac{(x-\hat{x})}{x_{\mathrm{max}}},
\end{equation}
where $\hat{x}$ is the estimated value of the variable from the approximation models and $x_{\mathrm{max}}$ is the maximum value from the dataset for each variables.
 The residual analysis shows that across all variables and flight phases, the residuals are centered around zero and exhibits no discernible bias, indicating that the surrogate models fo not systematically over or under predict the behavior of the system. Moreover, the absence of visible patterns in~\cref{fig:Tinres,fig:mfcres,fig:mfgt_res} suggests that the order of the system and nonlinearities are sufficiently captured.
\fi

\subsection{Problem Formulation}
The objective is to minimize the fuel consumption over a given time horizon $t\in[0, T]$ during each flight phase,
\begin{equation} \label{eq:obj}
    \min{{\int_{0}^{T}\dot{m}_{\mathrm{f}} \ \mathrm{d}t}},
\end{equation}
where $\dot{m}_{\mathrm{f}} = \dot{m}_{\mathrm{f,gt}} + \dot{m}_{\mathrm{f,fc}}$ denotes the combined fuel flow rate to the gas turbine and the fuel cell. The required propulsion power is enforced through the power balance in the integrated propulsion system,
\begin{equation}\label{eq:req}
P_{\mathrm{req}} = P_{\mathrm{gt}} + \eta \cdot \left(P_{\mathrm{fc}}-P_{\mathrm{aux}}\right).
\end{equation}
Here, the $P_{\mathrm{req}}$, $P_{\mathrm{aux}}$ and $\eta$ are the required power, auxiliary power, and electrical motor efficiency, respectively. These quantities are assumed to be known. 
Based on the quasi-static models in \ifjournal~\crefrange{eq:modelTO}{eq:ThpcTOC} and~\cref{eq:req}\else~\crefrange{eq:affine}{eq:req}\fi, together with the safety and operational constraints of the propulsion system, the optimal control problem is framed in Problem~\ref{prob:main}.
\begin{prob}[Optimal power management problem]\label{prob:main}
Given a power request, $P_{\mathrm{req}}$, the optimal power to be provided by GT and SOFC results from the solution of:
\end{prob}
\begin{equation} \label{eq:optp}
    \begin{aligned}
         \min_{P_{\mathrm{fc}}}\quad
         & {{\int_{0}^{T}\dot{m}_{\mathrm{f}} \ dt}}& \\
         \textnormal{s.t. } &\ifjournal \textnormal{\crefrange{eq:modelTO}{eq:ThpcTOC} \ and  \cref{eq:req}}
          \else \textnormal{\cref{eq:affine,eq:quadratic,eq:req}}, \fi & \\
          &\begin{aligned}
          P_{\mathrm{gt}} &\in [P_{\mathrm{gt,{\mathrm{min}}}},P_{\mathrm{gt,{\mathrm{max}}}}], \\
          P_{\mathrm{fc}} &\in [P_{\mathrm{fc,{\mathrm{min}}}},P_{\mathrm{fc,{\mathrm{max}}}}], \\
          P_{\mathrm{em}} &\in [P_{\mathrm{em,{\mathrm{min}}}},P_{\mathrm{em,{\mathrm{max}}}}], \\
	T_{\mathrm{in}} &\in [T_{\mathrm{in,{\mathrm{min}}}},T_{\mathrm{in,{\mathrm{max}}}}], \\
	T_{\mathrm{hpc}} &\in [T_{\mathrm{hpc,{\mathrm{min}}}},T_{\mathrm{hpc,{\mathrm{max}}}}], \\
	T_{\mathrm{et}} &\in [T_{\mathrm{et,{\mathrm{min}}}},T_{\mathrm{et,{\mathrm{max}}}}], \\
	T_{\mathrm{out}} &\in [T_{\mathrm{out,{\mathrm{min}}}},T_{\mathrm{out,{\mathrm{max}}}}], \\
	 \dot{m}_{\mathrm{f,gt}} &\in [\dot{m}_{\mathrm{f,gt,{\mathrm{min}}}},\dot{m}_{\mathrm{f,gt,{\mathrm{max}}}}], \\
          \dot{m}_{\mathrm{f,fc}} &\in [\dot{m}_{\mathrm{f,fc,{\mathrm{min}}}},\dot{m}_{\mathrm{f,fc,{\mathrm{max}}}}], \\
	 \dot{m}_{\mathrm{b}} &\in [\dot{m}_{\mathrm{b,{\mathrm{min}}}},\dot{m}_{\mathrm{b,{\mathrm{max}}}}]. 
	 \end{aligned}\\
    \end{aligned}
\end{equation}
The decision variable in this problem is the fuel cell output power $P_{\mathrm{fc}}$, such that the  Problem~\ref{prob:main} reduces to a one dimensional optimization problem. Because the system model is purely algebraic and does not include dynamics, the optimization in Problem~\ref{prob:main} reduces to a pointwise steady state problem. The resulting control action is therefore computed offline or online as a static optimal control input for each operating point. The~\cref{eq:obj} defines the total hydrogen consumption minimized by the optimal controller, while~\cref{eq:req} enforces the power balance between mechanical and electrical subsystems. The remaining constraints in~\cref{eq:optp} represent thermal, mechanical and operational limits. 
\ifjournal
\begin{lemma}[Monotonicity of SOFC inlet temperature] \label{lem:Tinmono}
	Let the SOFC inlet temperature be approximated by the quadratic term~\cref{eq:TinTO}
	. After substituting the power balance constraint~\cref{eq:req}, the mapping $T_{\mathrm{in}}(P_{\mathrm{fc}})$ is a concave quadratic function of $P_{\mathrm{fc}}$ with a unique maximizer at 
	\begin{equation}
		P_{\mathrm{fc}}^{\mathrm{peak}} = - \frac{b}{2a},
	\end{equation}
	where $a < 0$.
	If the admissible operating set satisfies 
	\begin{equation}
		P_{\mathrm{fc}} \le P_{\mathrm{fc}}^{\mathrm{peak}},
	\end{equation}
	then $T_{\mathrm{in}}(P_{\mathrm{fc}})$ is strictly increasing over the feasible set.
\end{lemma}
\begin{proof}
	Substituting the power balance into~\cref{eq:TinTO} 
	 yields
	\begin{equation}
		T_{\mathrm{in}}(P_{\mathrm{fc}}) = a P_{\mathrm{fc}}^2 + b P_{\mathrm{fc}} + c,
	\end{equation}	  
	where $a$, $b$ and $c$ are defined in~\crefrange{eq:aT}{eq:cT}. The coefficient $a < 0$ follows from $Q \preceq 0$. The derivative is
	\begin{equation}
		\frac{\mathrm{d} T_{\mathrm{in}}}{\mathrm{d} P_{\mathrm{fc}}} = 2a P_{\mathrm{fc}} + b.
	\end{equation}
	This derivative is positive for all $P_{\mathrm{fc}} < -\frac{b}{2a}$ and vanishes uniquely at
	\begin{equation}
		P_{\mathrm{fc}}^{\mathrm{peak}} = -\frac{b}{2a}.
	\end{equation}
	Therefore, $T_{\mathrm{in}} (P_{\mathrm{fc}})$ is strictly increasing over the admissible region. 
\end{proof}
By the~\cref{lem:Tinmono} the SOFC inlet temperature constraint $T_{\mathrm{in}} \leq T_{\mathrm{in,max}}$ can be equivalently expressed as an upper bound on $P_{\mathrm{fc}}$. Consequently, although the quadratic constraint yields two analytical roots, only the smaller root lies within the admissible region.
This consideration with the constraints in the Problem~\ref{prob:main}, not only ensure the performance safety of the system but also restrict the problem to the admissible range so that all admissible operating points remain in the reliable region.

\else
Although the SOFC inlet temperature surrogate is concave, its dependence on the fuel-cell power is monotonic over the admissible operating region. Consequently, the constraint $T_{\mathrm{in}} \leq T_{\mathrm{in,max}}$ is equivalent to an upper bound on $P_{\mathrm{fc}}$. This reformulation preserves convexity of the feasible set and allows the resulting optimization problem to be solved using convex optimization methods and guarantees optimality.

\fi
\section{Methodology} \label{sec:met}
To derive a closed-form supervisory power management policy, we formulate the Lagrangian of Problem~\ref{prob:main} and apply the KKT optimality conditions~\citep{rao2019}. For convex optimization problems, these conditions are both necessary and sufficient for optimality and consists of stationarity, primal feasibility, dual feasibility, and complementary slackness. Since no system dynamics are included, the resulting policy does not provide feedback adaptation but represents a computationally efficient feedforward control input evaluated independently at each operating point. 

Based on the Problem~\ref{prob:main}, the system model \ifjournal in~\crefrange{eq:modelTO}{eq:ThpcTOC} \else using~\crefrange{eq:affine}{eq:quadratic} \fi are expressed using affine, convex and monotonic concave functions together with the constraints on power, temperatures, and mass flow rate variables. Considering the monotonicity of the non-convex quadratic constraints, the resulting problems admits a unique solution in the feasible set. This structure results in explicit expressions for the optimal fuel cell power at each operating point without iterative numerical solvers.
\ifjournal
\ifjournal
\begin{theorem}[Structure of the optimal policy]\label{theorem:the1}
	For all admissible steady state operating points studied in this paper, the power management policy saturates the SOFC inlet temperature constraint whenever feasible. The corresponding optimal control policy is given by 
\begin{equation}
	P_{\mathrm{fc, const}}^* = \frac{-b + \sqrt{b^2-4ac}}{2a},
\end{equation}
where
\begin{align}
	a &= Q_{22} - 2Q_{12}\eta+Q_{11}\eta^2, \\
	b &= 2(Q_{12}-\eta Q_{11})(P_{\mathrm{req}} + \eta P_{\mathrm{aux}}) -q_{1}\eta+q_{2},\\
	c &= Q_{11}(P_{\mathrm{req}}+\eta P_{\mathrm{aux}})^2 + q_{1}(P_{\mathrm{req}}+\eta P_{\mathrm{aux}}) + \notag\\
	&\quad q_0 - T_{\mathrm{in,max}},
\end{align}
where $Q$, $\mathbf{q}$ and $q_0$ are the mapping coefficients of the $T_{\mathrm{in}}$ model.
\end{theorem}
\begin{proof}
	The unconstrained optimum corresponds to the minimum fuel consumption obtained by neglecting system constraints, which equals
\begin{equation}
	P_{\mathrm{fc,unc}}^* = \begin{cases}
		\begin{aligned}
				&\frac{2(Q_{\mathrm{f}11}\eta - Q_{\mathrm{f}12})(P_{\mathrm{req}}+\eta P_{\mathrm{aux}})+}{2(Q_{\mathrm{f}11}\eta^2 -2Q_{\mathrm{f}12}\eta } \\
				&\frac{ (b_1+ q_{\mathrm{f}1})\eta - (q_{\mathrm{f}2}
					+b_2)}{+ Q_{\mathrm{f}22})}
			\end{aligned} 
			& \text{if } \text{take off}, \\
		+\infty & \hspace{-60pt} \text{if } \text{top of climb and cruise}.
	\end{cases}
\end{equation}
However, feasibility is determined with respect to the operational and thermal constraints. 
Therefore, if $P_{\mathrm{fc,unc}}^*$ lies beyond the feasible set, the KKT conditions guarantee that the optimal solution occurs at an active constraint. 
\subsection*{Constrained Solution}
The KKT conditions associated with Problem~\ref{prob:main} are given as
\begin{align}
&\frac{\mathrm{\partial}L}{\mathrm{\partial}P_{\mathrm{fc}}} = 0, \label{eq:station}\\
&\mu_i \cdot g_i = 0, \label{eq:cs} \\
&\mu_i \geq 0, \label{eq:dual}\\  
&g_i \leq 0  \label{eq:primal}, \ i = 1, 2, 3, ..., 20,
\end{align}
where~\ref{eq:station} denotes the stationarity condition,~\ref{eq:cs} is the complementary slackness,~\ref{eq:dual} is the dual feasibility, and~\ref{eq:primal} is the primal feasibility and it is expanded in~\cref{eq:gexp}. 
We define the corresponding Lagrangian as
\begin{equation} \label{eq:L1}
	L = \dot{m}_{\mathrm{f,fc}}+\dot{m}_{\mathrm{f,gt}} + \bm{\mu}^{\top}\cdot\mathbf{g} ,
\end{equation}
where $\bm{\mu} \in $ denotes the vector of Lagrange multipliers associated with the inequality constraints. 
Applying the stationary condition with respect to the decision variable $P_{\mathrm{fc}}$ gives
\begin{equation}\label{eq:pL}
	\frac{\mathrm{\partial}L}{\mathrm{\partial}P_{\mathrm{fc}}} = \frac{\mathrm{\partial}(\dot{m}_{\mathrm{f,fc}}+\dot{m}_{\mathrm{f,gt}})}{\mathrm{\partial}P_{\mathrm{fc}}} + \bm{\mu}^{\top} \frac{\mathrm{\partial}\mathbf{g}
	}{\mathrm{\partial}P_{\mathrm{fc}}},
\end{equation}
which leads to a piecewise analytical expression for the optimal power management policy after implementing~\crefrange{eq:cs}{eq:primal}.
As a result, when only affine constraints are active, the optimal value follows the closed form expression of 
\begin{equation}\label{eq:lin}
	P_{\mathrm{fc, lin}}^* = \frac{x_{j,\mathrm{max/min}}- \theta_{j,0} - \theta_{j,1} (P_{\mathrm{req}}+\eta P_{\mathrm{aux}})}{-\theta_{j,1}\eta +\theta_{j,2}} ,
\end{equation}
and if a quadratic constraint is active then the optimal power for fuel cell is obtained from the scalar quadratic root as
\begin{equation}
	P_{\mathrm{fc,quad}}^* = \frac{-b_j \pm \sqrt{b_j^2-4a_jc_j}}{2a_j},
\end{equation}
with coefficients
\begin{align}
	a_j &= Q_{j,22} - 2Q_{j,12}\eta+Q_{j,11}\eta^2, \\
	b_j &= 2(Q_{j,12}-\eta Q_{j,11})(P_{\mathrm{req}} + \eta P_{\mathrm{aux}})-q_{j,1}\eta+q_{j,2},\\
	c_j &= Q_{j,11}(P_{\mathrm{req}}+\eta P_{\mathrm{aux}})^2 + q_{j,1}(P_{\mathrm{req}}+\eta P_{\mathrm{aux}}) + \notag\\
	&\quad q_{j,0} - x_{j,\mathrm{max/min}}.
\end{align}
Therefore, when a linear constraint is active, the optimal fuel cell power reduces to the affine expression in~\cref{eq:lin}. Similarly, when a quadratic constraint is active the optimal power is obtained by solving a scalar quadratic equation. By implementing this methodology as explained in appendix~\ref{sec:ap}, the optimal power allocation is represented in~\cref{eq:pfcto,eq:pfctoc} and the effective optimal control policy is
		\begin{equation}
			P_{\mathrm{fc,opt}}^{\mathrm{eff}} = \min{\{P_{\mathrm{fc,const}}, P_{\mathrm{fc,unc}}\}}.
		\end{equation}
It is concluded that the only constraint that is active in all three studied operating points is the thermal limit enforced by the SOFC inlet temperature upper bound.
\end{proof}
According to~\cref{theorem:the1}, the constrained optimal power allocation for take off, top of climb and cruise are bounded by the thermal limitation of the SOFC inlet. The results derived by KKT conditions shows a boundary-seeking behavior.
\fi
\else
The unconstrained optimum corresponds to the minimum fuel consumption obtained by neglecting system constraints, which equals to
\begin{equation}
	P_{\mathrm{fc,unc}}^* = \begin{cases}
		\frac{2(Q_{\mathrm{f}11} + Q_{\mathrm{f}12})(P_{\mathrm{req}}+\eta P_{\mathrm{aux}}) +(b_1- q_{\mathrm{f}1})\eta + (q_{\mathrm{f}2}-b_2)}{2(Q_{\mathrm{f}11}\eta^2 -2Q_{\mathrm{f}12}\eta + Q_{\mathrm{f}22})} & \text{if } \text{take off}, \\
		+\infty & \text{if } \text{top of climb and cruise}.
	\end{cases}
\end{equation}
However, feasibility is determined with respect to the operational and thermal constraints. Therefore, if $P_{\mathrm{fc,unc}}^*$ lies beyond the feasible set, KKT conditions guarantee that the optimal solution occurs at an active constraint. 
\fi
\ifjournal
\else
The KKT conditions associated with Problem~\ref{prob:main} are given as
\begin{align}
&\frac{\mathrm{\partial}L}{\mathrm{\partial}P_{\mathrm{fc}}} = 0, \label{eq:station}\\
&\mu_i \cdot g_i = 0, \label{eq:cs} \\
&\mu_i \geq 0, \label{eq:dual}\\  
&g_i \leq 0  \label{eq:primal},  i = 1, 2, 3, ..., 18,
\end{align}
where~\cref{eq:station} denotes the stationary condition,~\cref{eq:cs} is the complementary slackness,~\cref{eq:dual} is the dual feasibility, and~\cref{eq:primal} is the primal feasibility. 
The constraint vector $\mathbf{g}$ collects the operational boundaries on lifting variables as well as the powers and inputs.
We define the corresponding Lagrangian as
\begin{equation} \label{eq:L1}
	L = \dot{m}_{\mathrm{f,fc}}+\dot{m}_{\mathrm{f,gt}} + \mathbf{\mu}^{\top}\cdot\mathbf{g} 
\end{equation}
where $\mathbf{\mu}$ denotes the vector of Lagrange multipliers associated with the inequality constraints. 

Applying the stationary condition with respect to the decision variable $P_{\mathrm{fc}}$ gives
\begin{equation}\label{eq:pL}
	\frac{\mathrm{\partial}L}{\mathrm{\partial}P_{\mathrm{fc}}} = \frac{\mathrm{\partial}(\dot{m}_{\mathrm{f,fc}}+\dot{m}_{\mathrm{f,gt}})}{\mathrm{\partial}P_{\mathrm{fc}}} + \mathbf{\mu}^{\top} \frac{\mathrm{\partial}\mathbf{g}
	}{\mathrm{\partial}P_{\mathrm{fc}}},
\end{equation}
which leads to a piecewise analytical expression for the optimal power management policy after implementing~\cref{eq:cs}-~\cref{eq:primal}.
As a result, when only affine constraints are active, the optimal value follows the closed form expression of 
\begin{equation}\label{eq:lin}
	P_{\mathrm{fc, lin}}^* = \frac{x_{j,\mathrm{max/min}}- \theta_{j,0} - \theta_{j,1} (P_{\mathrm{req}}+\eta P_{\mathrm{aux}})}{-\theta_{j,1}\eta +\theta_{j,2}} 
\end{equation}
and if a quadratic constraint is active then the optimal power for fuel cell is obtained from the scalar quadratic root as
\begin{equation}
	P_{\mathrm{fc,quad}}^* = \frac{-b_j \pm \sqrt{b_j^2-4a_jc_j}}{2a_j},
\end{equation}
with coefficients
\begin{align}
	a_j &= Q_{j,22} - 2Q_{j,12}\eta+Q_{j,11}\eta^2, \\
	b_j &= 2(Q_{j,12}-Q_{j,11})(P_{\mathrm{req}} + \eta P_{\mathrm{aux}})-q_{j,1}\eta+q_{j,2},\\
	c_j &= Q_{j,11}(P_{\mathrm{req}}+\eta P_{\mathrm{aux}})^2 + q_{j,1}(P_{\mathrm{req}}+\eta P_{\mathrm{aux}}) + q_{j,0} - x_{j,\mathrm{max/min}}.
\end{align}
Therefore, when a linear constraint is active, the optimal fuel cell power reduces to the affine expression in~\cref{eq:lin}. Similarly, when a quadratic constraint is active, like the inlet temperature limit, the optimal power is obtained by solving a scalar quadratic equation.
\fi

\ifjournal
\else
Applying the KKT stationary condition
shows that the dominant active constraint is the SOFC inlet temperature upper bound in all three phases. Consequently, the corresponding optimal control policy is given by the inequality constraint of the $T_{\mathrm{in}}$ upper bound, which leads to  
\begin{equation}
	P_{\mathrm{fc, const}}^* = \frac{-b + \sqrt{b^2-4ac}}{2a},
\end{equation}
where
\begin{align}
	a &= Q_{22} - 2Q_{12}\eta+Q_{11}\eta^2, \\
	b &= 2(Q_{12}- Q_{11})(P_{\mathrm{req}} + \eta P_{\mathrm{aux}}) -q_{1}\eta+q_{2},\\
	c &= Q_{11}(P_{\mathrm{req}}+\eta P_{\mathrm{aux}})^2 + q_{1}(P_{\mathrm{req}}+\eta P_{\mathrm{aux}}) + q_0 - T_{\mathrm{in,max}},
\end{align}
where $Q$, $\mathbf{q}$ and $q_0$ are the mapping coefficients of the $T_{\mathrm{in}}$ model. Eventually, the \ifjournal effective\fi optimal policy is
\begin{equation}
	P_{\mathrm{fc,opt}}\ifjournal^{\mathrm{eff}}\fi = \min{\{P_{\mathrm{fc,const}}, P_{\mathrm{fc,unc}}\}}.
\end{equation}
The full calculations and methodology are described in the extended version of the paper~\citep{PakEtAl2026}.
\fi
\section{Results}\label{sec:result}
To validate the proposed analytical control policy, we compare the resulting power split with the result from the numerical NLP based integrated model~\citep{PerraEtAl2026}
, for all three phases and over a representative range of power demands. We show the comparison in~\cref{fig:power_split_all}, which presents the distribution of power between the gas turbine subsystem and the electrical components along a range of total power requirement during different flight phases. As the power requirement increases, the gas turbine power increases approximately linearly, while the electrical motor power remains almost constant in all three phases. Showing that the electrical motor power is primarily constrained by the SOFC inlet temperature rather than being affected by the total power demand. However, take off requires substantially higher power demand than cruise and top of climb and larger part of the power is provided by the gas turbine rather than the electrical power source. The analytical results closely matches the numerical results, which confirms that the simplified models that have been used for analytical solution are sufficient for the power management problem. 
\ifjournal
Figure~\ref{fig:fuel_all} shows the corresponding overall fuel consumption of the propulsion system. As the power requirement increases the fuel consumption is also rising linearly in cruise and top of climb phases. The fuel flow rate in take off shows a weakly convex monotonic increase, which was expected from the quadratic fuel flow rate map for GT. As the power requirement is higher in take off than the other phases, the fuel consumption is higher than in cruise and top of climb. Comparing the results from the analytical power management policy with the numerical optimal strategy, the NRMSEs are less than $0.7 \%$ for fuel consumption and less than $1.5 \%$ for power allocation, confirming that the analytical solution provides an accurate approximation over the relevant operating range. 


\ifjournal
\fi
\else
Figure~\ref{fig:fuel_all} shows the corresponding overall fuel consumption of the propulsion system. As the power requirement increases the fuel consumption is also rising linearly in cruise and top of climb phases. The fuel flow rate in take off shows a weakly convex monotonic increase, which was expected from the quadratic fuel flow rate map for GT. As the power requirement is higher in take off than the other phases, the fuel consumption is higher than in cruise and top of climb. Comparing the results from the analytical control policy with the numerical optimal strategy, the NRMSEs are less than $0.7 \%$ for fuel consumption and less than $1.5 \%$ for power allocation, confirming that the analytical solution provides an accurate approximation over the relevant operating range. 

\ifjournal
To validate the policy performance during all phases, we show the comparison between the powersplit result by the analytical policy with the numerical result with high fidelity models in~\cref{fig:combined}. The figure shows an acceptable precision in results for the desired power ranges for all three phases, despite the simple models that were utilized. At each operating point, the proposed static feedforward power management policy yields optimal power setpoints that closely match the numerical steady state optimal solution. 
%
%

The optimal powersplit in~\cref{fig:ps} depicts that the share of the electrical motor power 
during different flight phases all together.
In terms of computational efficiency, solving the optimization problem for $360$ different data points, the numerical optimization requires $214.23~\mathrm{s}$, whereas the analytical method requires $0.81~\mathrm{s}$. Both computations were performed in \textsc{MATLAB} on a computer equipped with a 12th Gen Intel\textsubscript{\textregistered} Core\textsuperscript{\texttrademark} i7-12700H (2.30~GHz) processor and a 16.0~GB of RAM, highlighting the computational advantage of the analytical approach. 
\fi
\fi
%
\begin{figure}[t]
	\begin{subfigure}[t]{0.48\textwidth}
		\includegraphics[width=\textwidth]{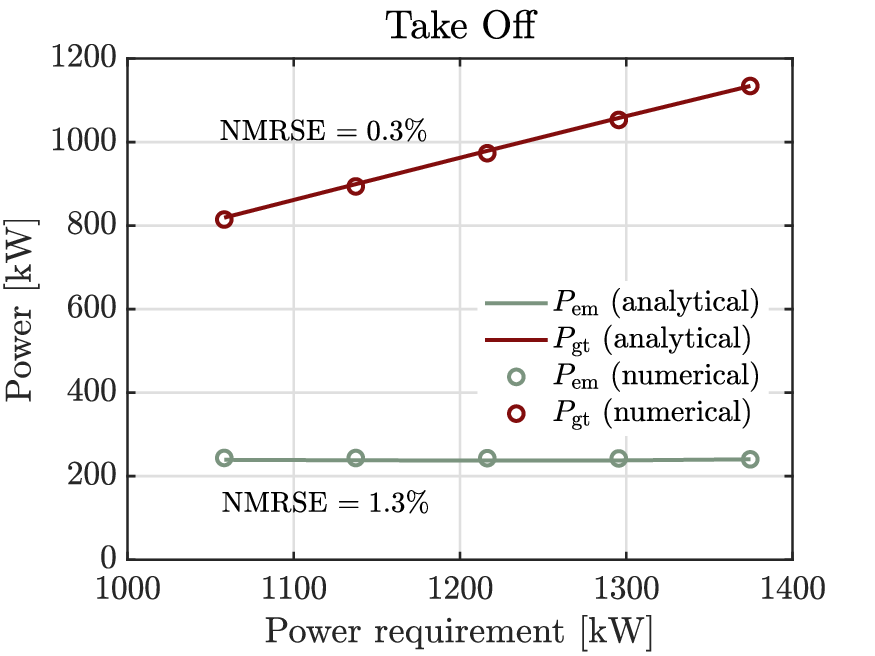}
		\caption{Take off}
		\label{fig:ps_to} 
	\end{subfigure} 
	\hfill
	\begin{subfigure}[t]{0.48\textwidth}
		\centering
		\includegraphics[width=\textwidth]{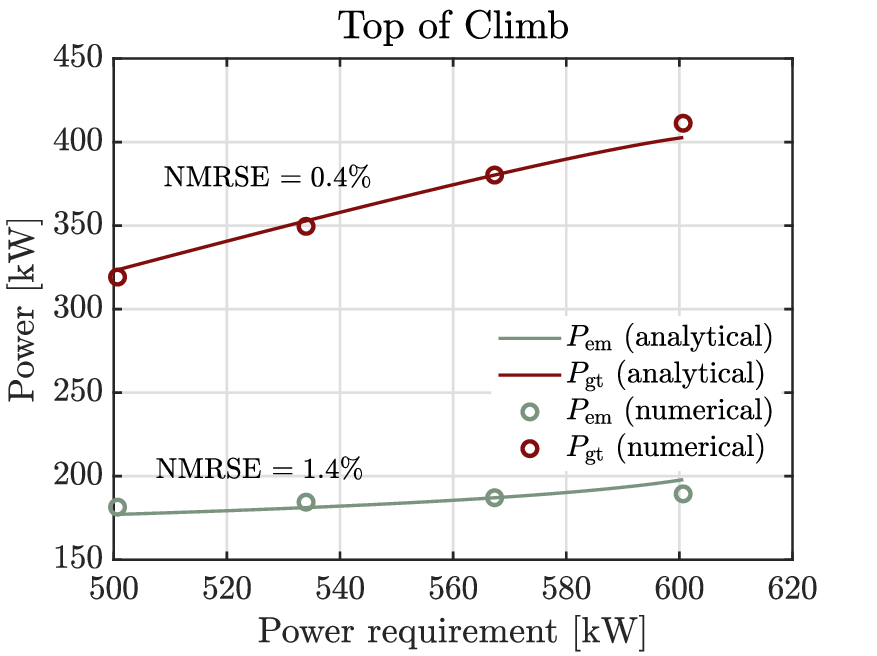}
		\caption{Top of climb}
		\label{fig:ps_toc}
	\end{subfigure} 
	\hfill
	\begin{subfigure}[t]{0.48\textwidth}
		\centering
		\includegraphics[width=\textwidth]{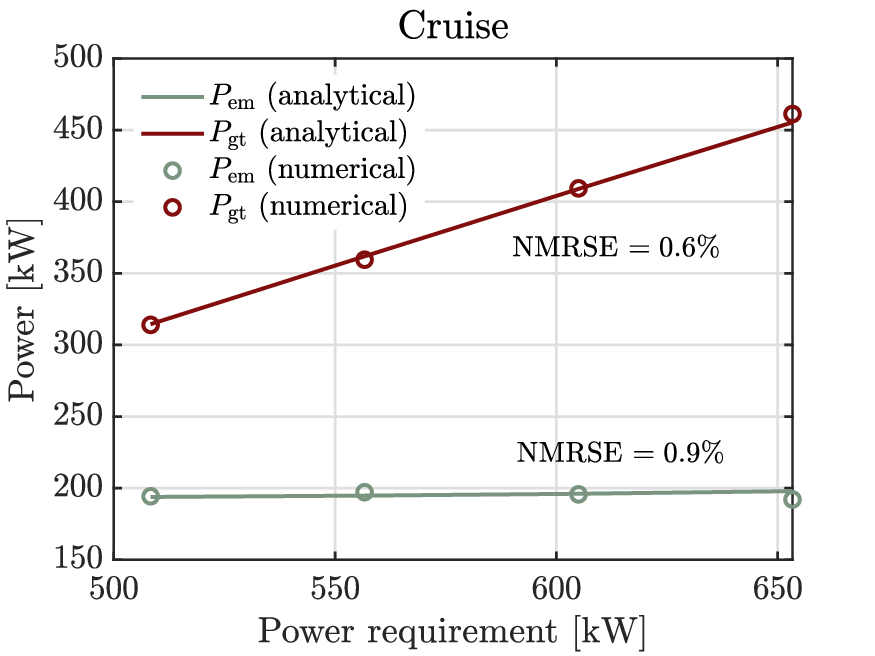}
		\caption{Cruise}
		\label{fig:ps_cr}
	\end{subfigure}
	
	\caption{Power split between subsystems for different power requirements across flight phases.}
	\label{fig:power_split_all}
\end{figure}
\begin{figure}[t]
	\centering
	\begin{subfigure}[t]{0.48\textwidth}
		\centering
		\includegraphics[width=\textwidth]{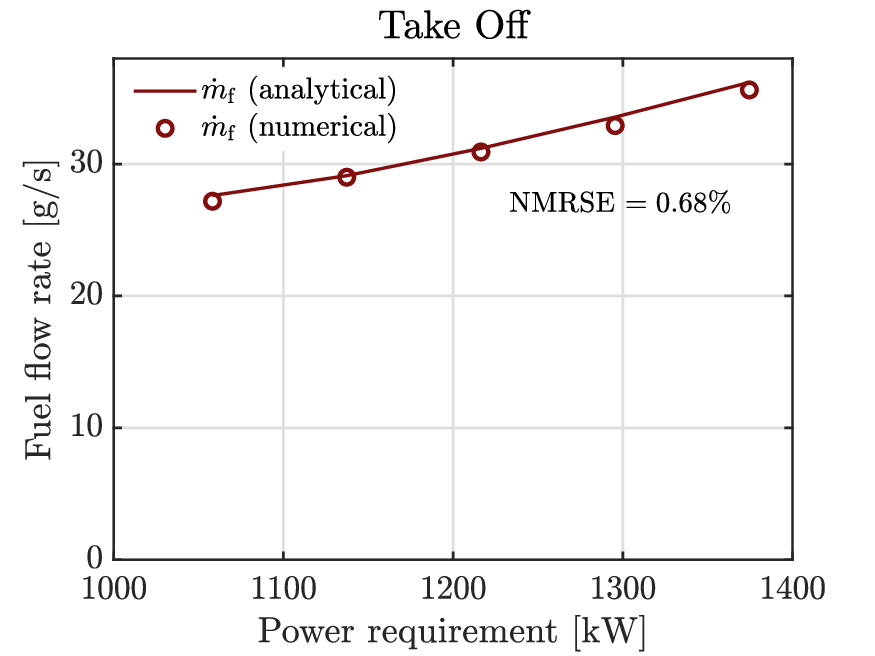}
		\caption{Take off}
		\label{fig:ps_to}
	\end{subfigure}
	\hfill
	\begin{subfigure}[t]{0.48\textwidth}
		\centering
		\includegraphics[width=\textwidth]{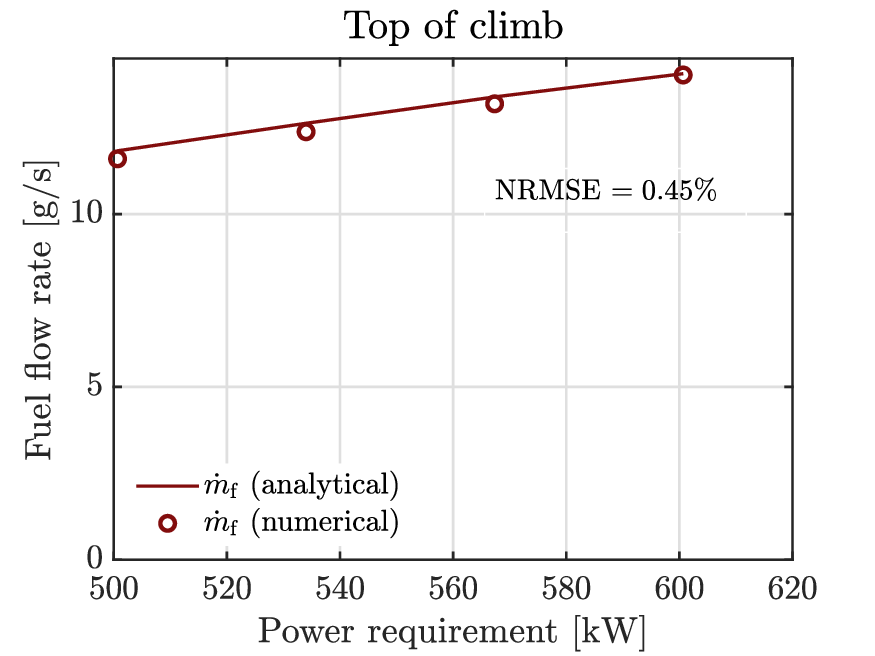}
		\caption{Top of climb}
		\label{fig:ps_toc}
	\end{subfigure}
	\hfill
	\begin{subfigure}[t]{0.48\textwidth}
		\centering
		\includegraphics[width=\textwidth]{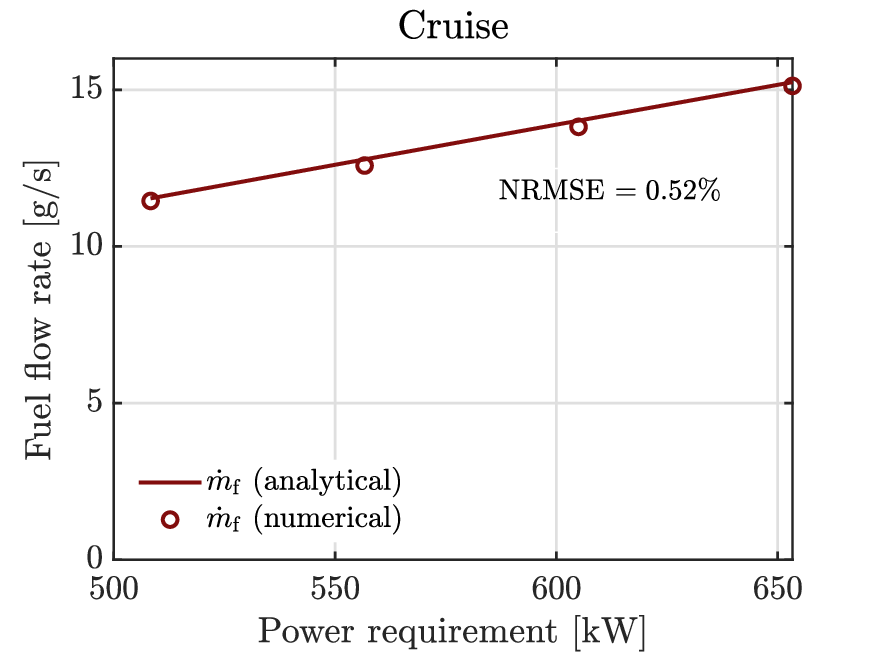}
		\caption{Cruise}
		\label{fig:ps_cr}
	\end{subfigure}
	
	\caption{Fuel flow rate for different power requirements across flight phases.}
	\label{fig:fuel_all}
\end{figure}
\begin{figure}[t]
	\centering
	\begin{subfigure}[t]{0.48\textwidth}
		\centering
		\includegraphics[width=\textwidth]{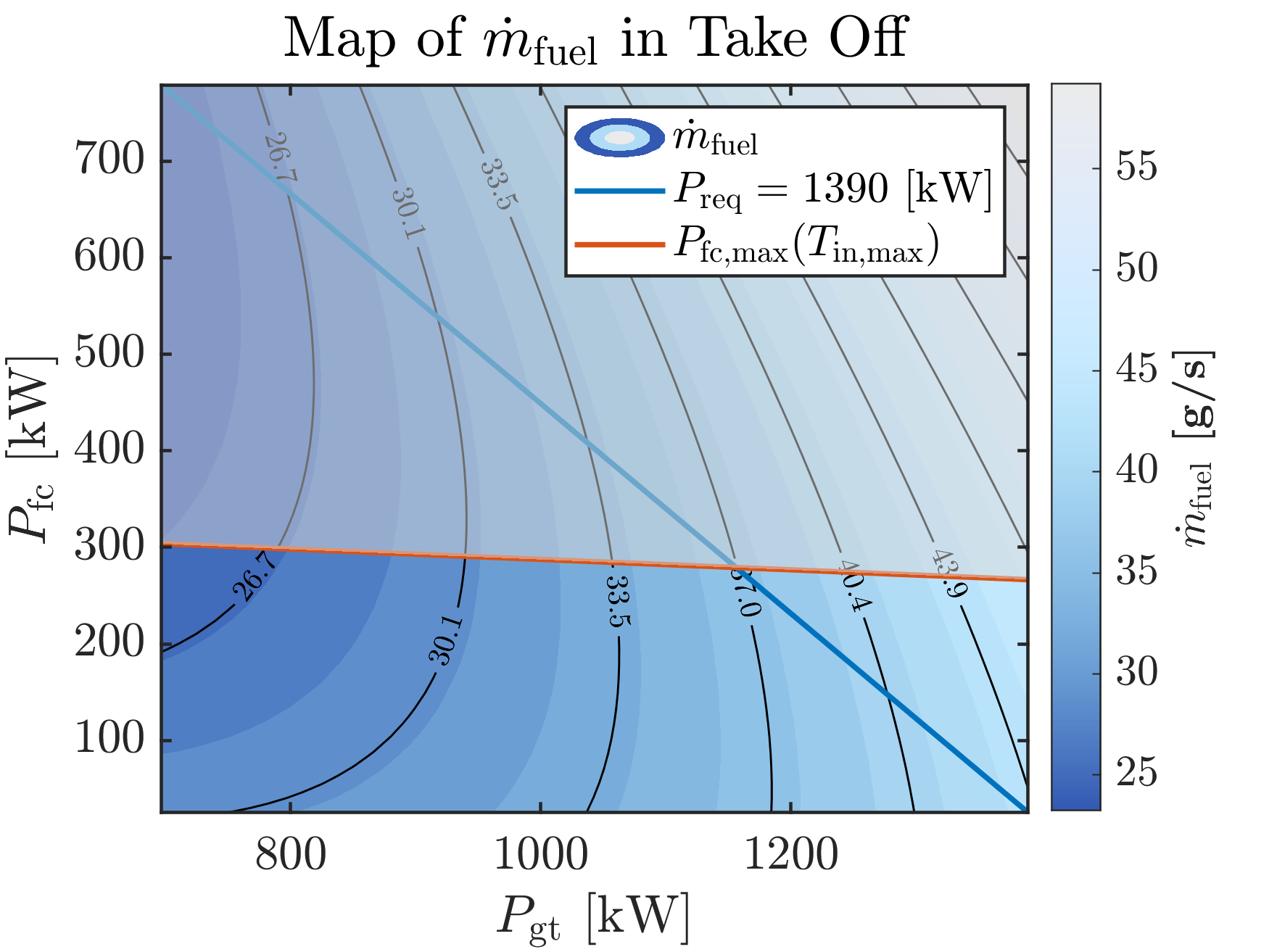}
		\caption{Take off}
		\label{fig:ps_to}
	\end{subfigure}
	\hfill
	\begin{subfigure}[t]{0.48\textwidth}
		\centering
		\includegraphics[width=\textwidth]{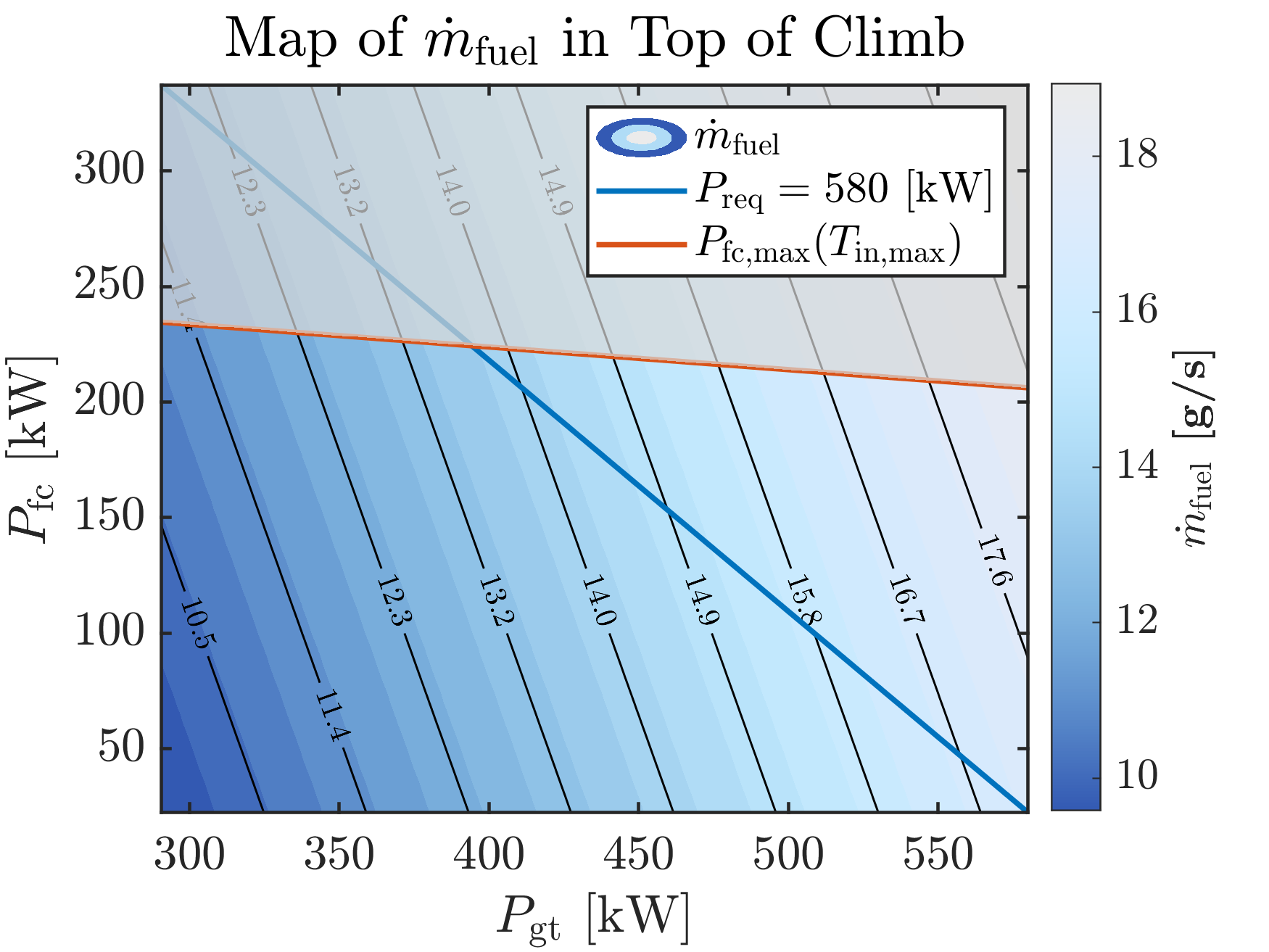}
		\caption{Top of climb}
		\label{fig:ps_toc}
	\end{subfigure}
	\hfill
	\begin{subfigure}[t]{0.48\textwidth}
		\centering
		\includegraphics[width=\textwidth]{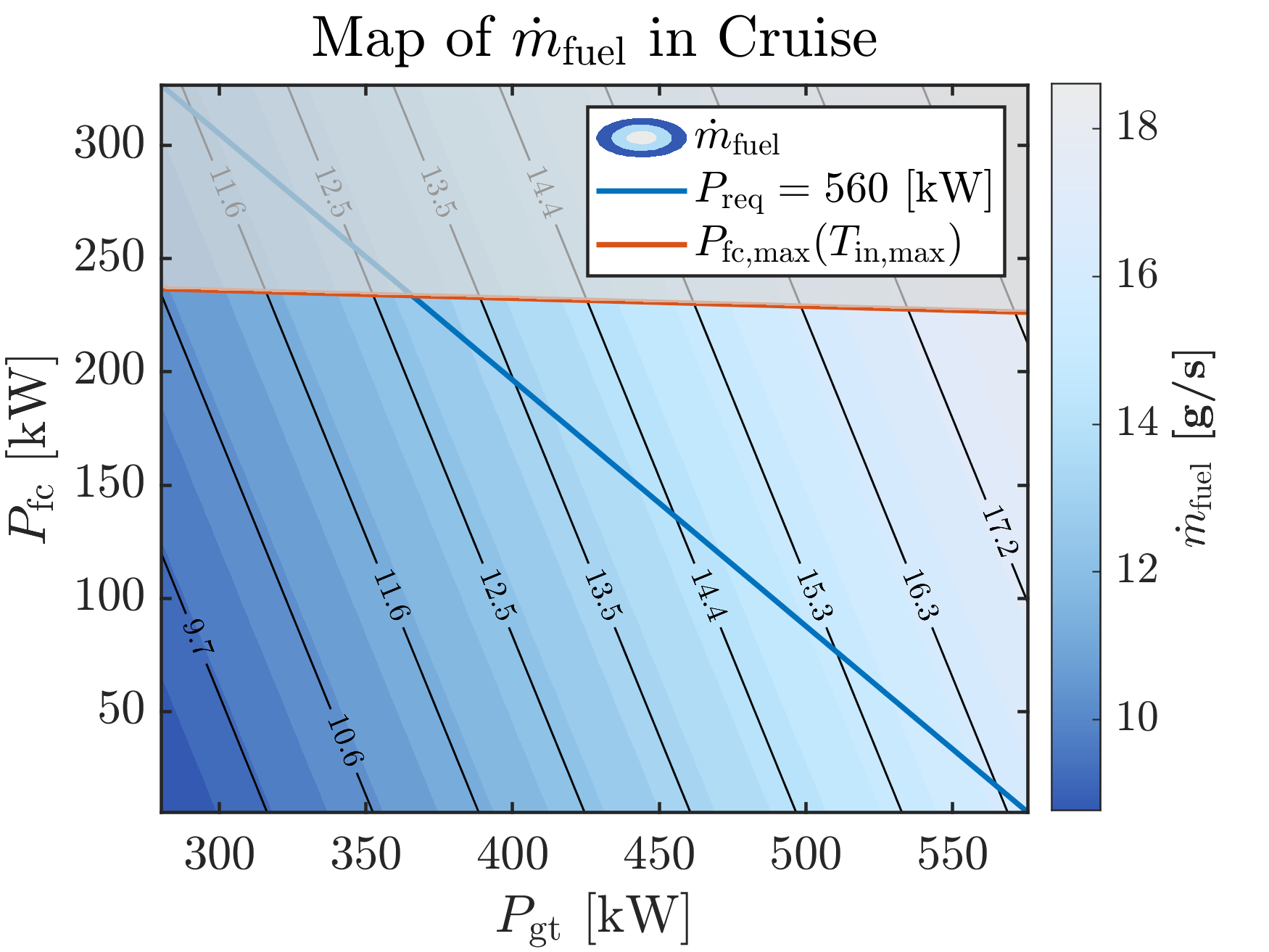}
		\caption{Cruise}
		\label{fig:ps_cr}
	\end{subfigure}
	
	\caption{Fuel flow rate contour maps for different flight phases. In each phase the constant power requirement is drawn in blue line to show how different power split affects the fuel consumption.}
	\label{fig:contour_all}
\end{figure}
\begin{figure}[t]
	\centering
	\vspace*{-5mm}
	\begin{subfigure}[t]{0.48\textwidth}
		\centering
		\includegraphics[width=\textwidth]{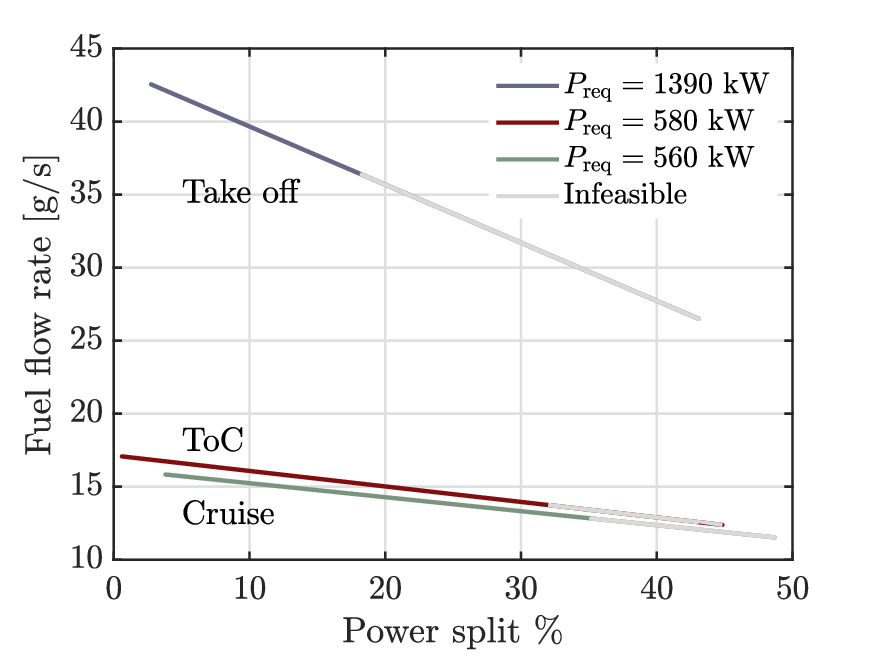}
		\label{fig:ps_mf}
	\end{subfigure}
	
	\caption{Fuel flow rate changes with respect to different power split schemes with three fixed power requirements.}
	\label{fig:ps_mfa}
\end{figure}
\begin{figure}[t]
	\centering
	\vspace*{-30mm}
	\begin{subfigure}[t]{0.48\textwidth}
		\centering
		\includegraphics[width=\textwidth]{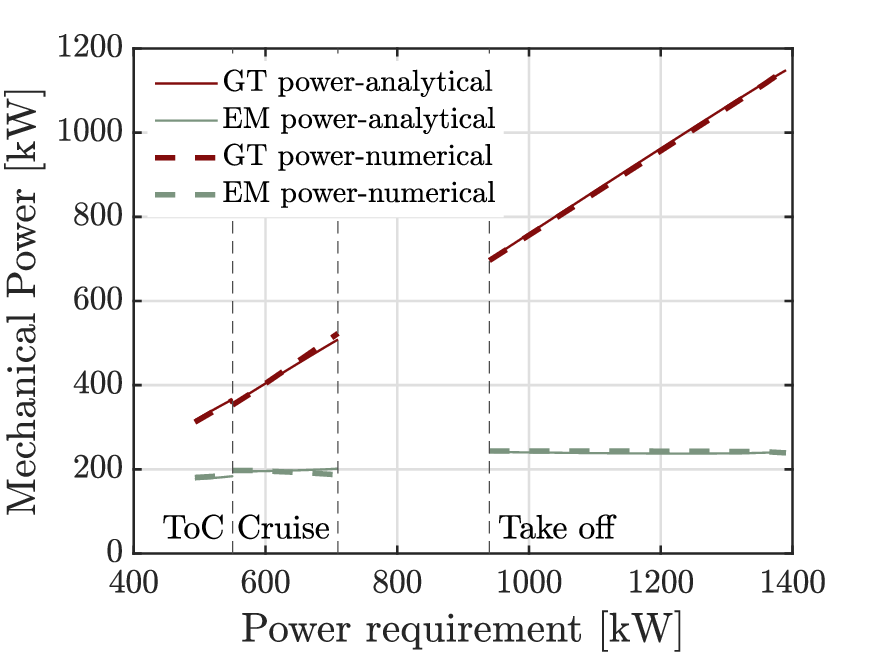}
		\caption{Power split for different power requirements}
		\label{fig:ps}
	\end{subfigure}
	\hfill
	\begin{subfigure}[t]{0.48\textwidth}
		\centering
		\includegraphics[width=\textwidth]{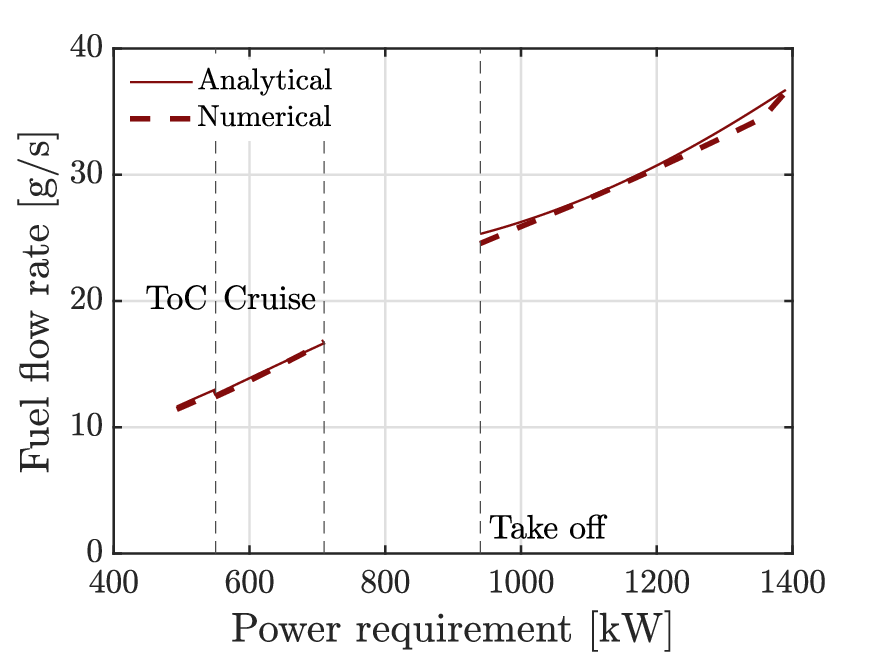}
		\caption{Fuel flow rate for different power requirements}
		\label{fig:ms}
	\end{subfigure}
	
	\caption{Comparison of power split and fuel flow rate for different power requirements.}
	\label{fig:combined}
	\vspace{-10pt}
\end{figure}
\ifjournal
%
The fuel flow rate maps in~\cref{fig:contour_all} illustrates the fuel consumption as function of gas turbine power and fuel cell power representing linear changes for top of climb and cruise and the quadratic behavior in take off. The maximum admissible fuel cell power is constraint by the fuel cell inlet temperature constraint and it is depicted as $P_{\mathrm{fc,max}}$. 

The contour plots show the feasible combination of gas turbine and fuel cell power that satisfy a given constant power requirement in each flight phase. The minimum fuel flow rates consistently occur along the upper bound for the fuel cell power resulting from the SOFC inlet temperature upper bound, which is interpreted that the fuel optimal operation is achieved by operating the SOFC at its maximum allowable power and supplying the remaining power demand with the gas turbine. This confirms that the optimal solution shows boundary seeking behavior with respect to the inlet temperature upper bound.

Moreover,~\cref{fig:ps_mfa} further illustrates the effect of power split on fuel consumption for three fixed total power requirements. The results show that increasing the fuel cell power consistently reduces the overall fuel flow rate, which is caused due to the higher efficiency of the SOFC compared to GT. However, this trend holds up to the maximum admissible fuel cell power, constrained by its inlet temperature. Any additional power demand beyond that is supplied by the GT.
To validate the policy performance during all phases, we show the comparison between the powersplit result by the analytical policy with the numerical result with high fidelity models in~\cref{fig:combined}. The figure shows an acceptable precision in results for the desired power ranges for all three phases, despite the simple models that were utilized. At each operating point, the proposed static feedforward power management policy yields optimal power setpoints that closely match the numerical steady state optimal solution. 
%

The optimal powersplit in~\cref{fig:ps} depicts that the share of the electrical motor power 
during different flight phases all together.
In terms of computational efficiency, solving the optimization problem for $360$ different data points, the numerical optimization requires $214.23~\mathrm{s}$, whereas the analytical method requires $0.81~\mathrm{s}$. Both computations were performed in \textsc{MATLAB} on a computer equipped with a 12th Gen Intel\textsubscript{\textregistered} Core\textsuperscript{\texttrademark} i7-12700H (2.30~GHz) processor and a 16.0~GB of RAM, highlighting the computational advantage of the analytical approach. 
\else
The fuel flow rate maps in~\cref{fig:contour_all} illustrates the fuel consumption as function of gas turbine power and fuel cell power representing linear changes for top of climb and cruise and the quadratic behavior in take off. The maximum admissible fuel cell power is constraint by the fuel cell inlet temperature constraint and it is depicted as $P_{\mathrm{fc,max}}$. The contour plots show the feasible combination of gas turbine and fuel cell power that satisfy a given constant power requirement in each flight phase. The minimum fuel flow rates consistently occur along the upper bound for the fuel cell power resulting from the SOFC inlet temperature upper bound, which is interpreted that the fuel optimal operation is achieved by operating the SOFC at its maximum allowable power and supplying the remaining power demand with the gas turbine. This confirms that the optimal solution shows boundary seeking behavior with respect to the inlet temperature upper bound.

Moreover,~\cref{fig:ps_mfa} further illustrates the effect of power split on fuel consumption for three fixed total power requirements. The results show that increasing the fuel cell power consistently reduces the overall fuel flow rate, which is caused due to the higher efficiency of the SOFC compared to GT. However, this trend holds up to the maximum admissible fuel cell power, constrained by its inlet temperature. Any additional power demand beyond that is supplied by the GT.

To validate the policy performance during all phases, we show the comparison between the powersplit result by the analytical policy with the numerical result with high fidelity models in~\cref{fig:combined}. The figure shows an acceptable precision in results for the desired power ranges for all three phases, despite the simple models that were utilized. At each operating point, the proposed static feedforward power management policy yields optimal power setpoints that closely match the numerical steady state optimal solution. 
%
%

The optimal powersplit in~\cref{fig:ps} depicts that the share of the electrical motor power 
during different flight phases all together.
In terms of computational efficiency, solving the optimization problem for $360$ different data points, the numerical optimization requires $214.23~\mathrm{s}$, whereas the analytical method requires $0.81~\mathrm{s}$. Both computations were performed in \textsc{MATLAB} on a computer equipped with a 12th Gen Intel\textsubscript{\textregistered} Core\textsuperscript{\texttrademark} i7-12700H (2.30~GHz) processor and a 16.0~GB of RAM, highlighting the computational advantage of the analytical approach. 
\fi
\section{Discussion and Future Work}\label{sec:dis}
%
While the current work establishes a theoretical foundation for the analytical optimal power management policy in the hybrid propulsion system, future research will focus on implementing and validating this policy under steady state operating conditions on a test-bench.
\section{Conclusion} \label{sec:con}
This paper presented a minimum-fuel supervisory power management policy for a hydrogen based hybrid aero engine combining a gas turbine and a solid oxide fuel cell. This policy enables optimal power split decisions without iterative computation.  We showed that the optimization problem can be formulated such that it results to a unique solution through curvature characteristics of the variables. Therefore, it can be solved analytically under steady state condition within the admissible domain. 
By devising KKT conditions to solve the optimization problem, a closed-form power management policy is obtained, which requires no iterative numerical computation while maintaining accuracy. The results revealed a key insight, which indicates that the resulted minimum-fuel behavior consistently drives the system to operate the SOFC at its thermal limit, highlighting the important role of thermal conditioning and SOFC inlet temperature constraint in such a hybrid propulsion system's efficiency. This outcome provides a physical interpretation of the power management decisions.  

The main implication is that the supervisory policy for hybrid aero engine could be reduced to a heuristic rule with lower computational complexity. 
The validation against a high-fidelity nonlinear model demonstrates lower than $1.5\%$ NRMSE, which means that the analytical solution is aligned with the numerical power management strategy during the representative flight phases. Additionally, this policy is simple and fast to implement on-line and ensures optimality within the admissible range, supporting that it is suitable for a real-time supervisory control. Since the derived close-form policy enables computationally efficient real-time power management, it is a practical solution for on-board implementation in a hybrid SOFC/GT propulsion system.
This method is limited to the steady-state operation; however, the proposed framework provides a foundation for 
future extensions toward dynamic optimal control for transient maneuvers using optimal control methods such as model predictive control (MPC). 

\ifjournal

\fi

\section*{Acknowledgment}
In this paper, we used Microsoft Copilot as an editor tool.



\clearpage 



\appendix
\numberwithin{equation}{section}
\renewcommand{\thesubsection}{\thesection.\arabic{subsection}}
\appendix
\section{Appendix} \label{sec:ap}
The constrained solution is obtained by solving the KKT conditions~\crefrange{eq:cs}{eq:primal}. The constraint vector $\mathbf{g}$ in~\cref{eq:primal} collects the operational boundaries on lifting variables as well as the powers and inputs,
\begin{equation}\label{eq:gexp}
	\mathbf{g} = \begin{bmatrix}
		P_{\mathrm{gt,min}} - P_{\mathrm{gt}}\\
		P_{\mathrm{gt}} - P_{\mathrm{gt,max}} \\ P_{\mathrm{fc,min}} - P_{\mathrm{fc}}\\
		P_{\mathrm{fc}} - P_{\mathrm{fc,max}} \\ P_{\mathrm{em,min}} - P_{\mathrm{em}}\\
		P_{\mathrm{em}} - P_{\mathrm{em,max}} \\ T_{\mathrm{in,min}}-T_{\mathrm{in}} \\ 
		T_{\mathrm{in}} - T_{\mathrm{in,max}} \\ T_{\mathrm{out,min}}-T_{\mathrm{out}} \\ 
		T_{\mathrm{out}} - T_{\mathrm{out,max}} \\
		\dot{m}_{\mathrm{f,fc, min}} - \dot{m}_{\mathrm{f,fc}} \\ \dot{m}_{\mathrm{f,fc}} - \dot{m}_{\mathrm{f,fc, max}}\\
		\dot{m}_{\mathrm{f,gt, min}} - \dot{m}_{\mathrm{f,gt}} \\ \dot{m}_{\mathrm{f,gt}} - \dot{m}_{\mathrm{f,gt, max}} \\
		\dot{m}_{\mathrm{b, min}} - \dot{m}_{\mathrm{b}} \\ \dot{m}_{\mathrm{b}} - \dot{m}_{\mathrm{b, max}} \\
		T_{\mathrm{hpc, min}} - T_{\mathrm{hpc}} \\ T_{\mathrm{hpc}} - T_{\mathrm{hpc, max}} \\
		T_{\mathrm{et, min}} - T_{\mathrm{et}} \\ T_{\mathrm{et}} - T_{\mathrm{et, max}} \\
        		 \end{bmatrix}.
\end{equation}
We expand the Lagrangian in~\cref{eq:L1} according to~\cref{eq:gexp}, which yields to
\begin{equation} \label{eq:h}
\begin{aligned}
   {L}=&\dot{m}_{\mathrm{f,fc}}+\dot{m}_{\mathrm{f,gt,tank}}+\mu_1(P_{\mathrm{gt,min}} - P_{\mathrm{gt}}) + \\ 
   &\mu_2(P_{\mathrm{gt}} - P_{\mathrm{gt,max}}) +  \mu_3(P_{\mathrm{fc,min}} - P_{\mathrm{fc}}) + \\ &\mu_4(P_{\mathrm{fc}} - P_{\mathrm{fc,max}}) +\mu_5(P_{\mathrm{em,min}} - P_{\mathrm{em}}) + \\ &\mu_6(P_{\mathrm{em}} - P_{\mathrm{em,max}}) + \mu_7(T_{\mathrm{in,min}}-T_{\mathrm{in}}) +  \\ &\mu_8(T_{\mathrm{in}} - T_{\mathrm{in,max}}) +\mu_9(\dot{m}_{\mathrm{f,fc, min}} - \dot{m}_{\mathrm{f,fc}}) +\\  
   &\mu_{10}(\dot{m}_{\mathrm{f,fc}} - \dot{m}_{\mathrm{f,fc, max}}) + \mu_{11}(\dot{m}_{\mathrm{f,gt, min}} - \dot{m}_{\mathrm{f,gt}})  + \\
   &\mu_{12}(\dot{m}_{\mathrm{f,gt}} - \dot{m}_{\mathrm{f,gt, max}}) +\mu_{13}(\dot{m}_{\mathrm{b, min}} - \dot{m}_{\mathrm{b}}) + \\ 
   &\mu_{14}(\dot{m}_{\mathrm{b}} - \dot{m}_{\mathrm{b, max}}) + \mu_{15}(T_{\mathrm{hpc, min}} - T_{\mathrm{hpc}}) +  \\ 
   &\mu_{16}(T_{\mathrm{hpc}} - T_{\mathrm{hpc, max}}) + \mu_{17}(T_{\mathrm{et, min}} - T_{\mathrm{et}}) + \\ 
   &\mu_{18}(T_{\mathrm{et}} - T_{\mathrm{et, max}})+  \mu_{19}(T_{\mathrm{out,min}}-T_{\mathrm{out}}) +\\ 
   &\mu_{20}(T_{\mathrm{out}} - T_{\mathrm{out,max}}).
\end{aligned}
\end{equation}
However, we can shorten the problem size in order to decrease the offline computational effort for each operating point.
\subsection{Take off}
During the take off phase, we combine the affine constraints into effective upper and lower bounds on the fuel cell power, denoted as $P_{\mathrm{fc,max}}$ and $P_{\mathrm{fc,min}}$, respectively, as
\begin{equation}
	\begin{aligned}
		P_{\mathrm{fc,max}}^{\mathrm{eff}} = &
		\min\Bigg\{
		P_{\mathrm{fc,max}},\;
		\frac{\left(P_{\mathrm{req}}+\eta P_{\mathrm{aux}} - P_{\mathrm{gt,min}}\right)}{\eta},\; \\
		&\frac{\dot{m}_{\mathrm{f,fc,max}} - b_0 - b_1(P_{\mathrm{req}}+\eta P_{\mathrm{aux}})}{-b_1\eta + b_2},\\[6pt]
		&\frac{T_{\mathrm{out,max}} - h_0 - h_1(P_{\mathrm{req}}+\eta P_{\mathrm{aux}})}{-h_1\eta + h_2},\; \frac{P_{\mathrm{em,max}}}{\eta}\\
		& + P_{\mathrm{aux}},\; \frac{T_{\mathrm{hpc,max}} - t_0 - t_1(P_{\mathrm{req}}+\eta P_{\mathrm{aux}})}{-t_1\eta + t_2}
		\Bigg\}
	\end{aligned}
\end{equation} 
\begin{equation}
	\begin{aligned}
		P_{\mathrm{fc,min}}^{\mathrm{eff}} = &
		\max\Bigg\{
		P_{\mathrm{fc,min}},\;
		\frac{(P_{\mathrm{req}}+\eta P_{\mathrm{aux}} - P_{\mathrm{gt,max}})}{\eta},\; \\
		&\frac{\dot{m}_{\mathrm{f,fc,min}} - b_0 - b_1(P_{\mathrm{req}}+\eta P_{\mathrm{aux}})}{-b_1\eta + b_2},\\[6pt]
		&\frac{T_{\mathrm{out,min}} - h_0 - h_1(P_{\mathrm{req}}+\eta P_{\mathrm{aux}})}{-h_1\eta + h_2},\; \frac{P_{\mathrm{em,min}}}{\eta} \\
		&+ P_{\mathrm{aux}},\; \frac{T_{\mathrm{hpc,min}} - t_0 - t_1(P_{\mathrm{req}}+\eta P_{\mathrm{aux}})}{-t_1\eta + t_2}
		\Bigg\}.
	\end{aligned}
\end{equation}
This bounds are obtained by enforcing the power requirements, the power balance~\cref{eq:req}, electrical components limits, fuel flow rate constraints and temperature bounds.
By using the effective bounds, we reduce the original constraint set in Problem~\ref{prob:main} to a smaller set of constraints presented in $\mathbf{g}^{\mathrm{eff}}$ as
\begin{equation}
	\mathbf{g}^{\mathrm{eff}} = \begin{bmatrix}
		P_{\mathrm{fc,min}}^{\mathrm{eff}} - P_{\mathrm{fc}}\\
		P_{\mathrm{fc}} - P_{\mathrm{fc,max}}^{\mathrm{eff}}  \\ T_{\mathrm{in,min}}-T_{\mathrm{in}}(P_{\mathrm{fc}}) \\ 
		T_{\mathrm{in}}(P_{\mathrm{fc}}) - T_{\mathrm{in,max}} \\\dot{m}_{\mathrm{b,min}}-\dot{m}_{\mathrm{b}}(P_{\mathrm{fc}}) \\ 
		\dot{m}_{\mathrm{b}}(P_{\mathrm{fc}}) - \dot{m}_{\mathrm{b,max}}\\\dot{m}_{\mathrm{f,gt,min}}-\dot{m}_{\mathrm{f,gt}}(P_{\mathrm{fc}}) \\ 
		\dot{m}_{\mathrm{f,gt}}(P_{\mathrm{fc}}) - \dot{m}_{\mathrm{f,gt,max}} 
	\end{bmatrix}
\end{equation}
This reduction simplifies the KKT analysis while preserving all critical operational limits. The associated Lagrangian is written as
\begin{equation} \label{eq:h}
	\begin{aligned}
		{L}= &\dot{m}_{\mathrm{f,fc}}+\dot{m}_{\mathrm{f,gt}}+\mu^{\mathrm{eff}}_1(P_{\mathrm{fc,min}}^{\mathrm{eff}} - P_{\mathrm{fc}}) +  \\ 
		&\mu^{\mathrm{eff}}_2(P_{\mathrm{fc}} - P_{\mathrm{fc,max}}^{\mathrm{eff}}) + \mu^{\mathrm{eff}}_3(T_{\mathrm{in,min}}-T_{\mathrm{in}}) + \\
		&\mu^{\mathrm{eff}}_4(T_{\mathrm{in}} - T_{\mathrm{in,max}})+ \mu^{\mathrm{eff}}_{5}(\dot{m}_{\mathrm{b, min}} - \dot{m}_{\mathrm{b}}) + \\ &\mu^{\mathrm{eff}}_{6}(\dot{m}_{\mathrm{b}} - \dot{m}_{\mathrm{b, max}}) + \mu^{\mathrm{eff}}_{7}(\dot{m}_{\mathrm{f,gt, min}} - \dot{m}_{\mathrm{f,gt}})+ \\ &\mu^{\mathrm{eff}}_{8}(\dot{m}_{\mathrm{f,gt}} - \dot{m}_{\mathrm{f,gt, max}}) ),
	\end{aligned}
\end{equation}
where $\mu^{\mathrm{eff}}_i$ denotes the Lagrange multiplier in the reduced KKT problem. The cost function for the take-off is
\begin{equation}
	\begin{aligned}
		J = &\dot{m}_{\mathrm{f}} = b_0 + b_1\cdot (P_{\mathrm{req}} - \eta(P_{\mathrm{fc}} - P_{\mathrm{aux}})) + b_2\cdot P_{\mathrm{fc}} + \\ 
		&\left(Q_{\mathrm{f}11}\eta^2 - 2Q_{\mathrm{f}12}\eta+Q_{\mathrm{f}22}\right) P_{\mathrm{fc}}^2 - \\   
		&\left(2(Q_{\mathrm{f}12}-Q_{\mathrm{f}11}\eta)(P_{\mathrm{req}} + \eta P_{\mathrm{aux}}) -q_{\mathrm{f}1}\eta+q_{\mathrm{f}2}\right)\\ 
		&P_{\mathrm{fc}} +  Q_{\mathrm{f}11}(P_{\mathrm{req}}+\eta P_{\mathrm{aux}})^2 +\\ &q_{\mathrm{f}1}(P_{\mathrm{req}}+\eta P_{\mathrm{aux}}) + q_{\mathrm{f,0}}.
	\end{aligned}
\end{equation}
Solving unconstrained problem, $\frac{\partial J}{\partial P_{\mathrm{fc}}} = 0 $, yields the unconstrained optimal solution
\begin{equation}
	\begin{aligned}
	P_{\mathrm{fc,unc}}^* = &\frac{2(Q_{\mathrm{f}11}\eta- Q_{\mathrm{f}12})(P_{\mathrm{req}}+\eta P_{\mathrm{aux}})}{2(Q_{\mathrm{f}11}\eta^2 -2Q_{\mathrm{f}12}\eta} \\ 
	&\frac{+(b_1+ q_{\mathrm{f}1})\eta - (q_{\mathrm{f}2}+b_2)}{+ Q_{\mathrm{f}22})}.
	\end{aligned}
\end{equation}
However, according to KKT condition the constrained optimal policy is achieved by adapting~\cref{eq:pL} to
\begin{equation}\label{eq:pL22}
	\frac{\mathrm{\partial}L}{\mathrm{\partial}P_{\mathrm{fc}}} = \frac{\mathrm{\partial}(\dot{m}_{\mathrm{f,fc}}+\dot{m}_{\mathrm{f,gt}})}{\mathrm{\partial}P_{\mathrm{fc}}} + \mathbf{\bm{\mu}^{\mathrm{eff}}}^{\top} \cdot \frac{\mathrm{\partial}\mathbf{g}^{\mathrm{eff}}}{\mathrm{\partial}P_{\mathrm{fc}}}.
\end{equation}
Solving~\cref{eq:pL22} reveals that the power requirement is distributed between the subsystems, pushing SOFC to produce power as much as possible, while it is constrained by its inlet temperature upper bound, i.e., $T_{\mathrm{in}}\leq T_{\mathrm{in,max}}$. Consequently, the optimal control policy is given by the constrained solution when $g^{\mathrm{eff}}_4 = 0$ is active, equals to  
\begin{equation} \label{eq:pfcto}
	P_{\mathrm{fc, const}}^* = \frac{-b + \sqrt{b^2-4ac}}{2a},
\end{equation}
where the $a$, $b$ and $c$ coefficients are
\begin{align}
	a &= Q_{22} - 2Q_{12}\eta+Q_{11}\eta^2, \label{eq:aT} \\
	b &= 2(Q_{12}- \eta Q_{11})(P_{\mathrm{req}} + \eta P_{\mathrm{aux}}) -q_{1}\eta+q_{2}, \label{eq:bT}\\
	c &= Q_{11}(P_{\mathrm{req}}+\eta P_{\mathrm{aux}})^2 + q_{1}(P_{\mathrm{req}}+\eta P_{\mathrm{aux}}) + \notag\\
	&\quad q_0 - T_{\mathrm{in,max}}.\label{eq:cT}
\end{align}
To calculate the active constraint's corresponding Lagrange multiplier we substitute the optimal fuel cell power in the stationarity condition as
\begin{equation}
	\begin{aligned}
		0 = &\frac{\partial L}{\partial P_{\mathrm{fc}}} = -\eta b_1 + b_2 + 2(Q_{\mathrm{f}11}\eta^2 - 2Q_{\mathrm{f}12}\eta+Q_{\mathrm{f}22})\\ 
		&P_{\mathrm{fc,opt}}^{\mathrm{eff}} - 2Q_{\mathrm{f}11}\eta(P_{\mathrm{req}} + \eta P_{\mathrm{aux}}) +2Q_{\mathrm{f}12}(P_{\mathrm{req}}+\eta P_{\mathrm{aux}}) - \\ &q_{\mathrm{f}1}\eta+q_{\mathrm{f}2} + \mu_4^{\mathrm{eff}}( 2(Q_{11}\eta^2 -2 Q_{12}\eta+Q_{22})P_{\mathrm{fc,opt}}^{\mathrm{eff}} + \\ 
		&2(Q_{12}-Q_{11}\eta)(P_{\mathrm{req}} + \eta P_{\mathrm{aux}})-q_{1}\eta+q_{2}).
	\end{aligned}
\end{equation}
Therefore the dual multiplier $\mu_4^{\mathrm{eff}}$ is calculated as
\begin{equation}
	\begin{aligned}
	\mu_4^{\mathrm{eff}} = &\frac{\eta b_1 - b_2 - 2(Q_{\mathrm{f}11}\eta^2 - 2Q_{\mathrm{f}12}\eta+Q_{\mathrm{f}22})P_{\mathrm{fc,opt}}^{\mathrm{eff}} }{2(Q_{11}\eta^2 - 2Q_{12}\eta+Q_{22})P_{\mathrm{fc,opt}}^{\mathrm{eff}} + } \\ 
	&\frac{+2( Q_{\mathrm{f}11}\eta-Q_{\mathrm{f}12})(P_{\mathrm{req}} + \eta P_{\mathrm{aux}}) +q_{\mathrm{f}1}\eta-q_{\mathrm{f}2}}{2(Q_{12}-Q_{11}\eta)(P_{\mathrm{req}} + \eta P_{\mathrm{aux}})-q_{1}\eta+q_{2}},
	\end{aligned}
\end{equation}
which represents its dependency on the power requirement $P_{\mathrm{req}}$ and the optimal fuel cell power $P_{\mathrm{fc,opt}}^{\mathrm{eff}}$.

\subsection{Top of climb and cruise}
As in take off case, we combine the affine constraints into effective upper and lower bounds on the fuel cell power as
\begin{equation}
	\begin{aligned}
		P_{\mathrm{fc,max}}^{\mathrm{eff}} = &
		\min\Bigg\{
		P_{\mathrm{fc,max}},\;
		\frac{\left(P_{\mathrm{req}}+\eta P_{\mathrm{aux}} - P_{\mathrm{gt,min}}\right)}{\eta},\; \\
		&\frac{\dot{m}_{\mathrm{f,fc,max}} - b_0 - b_1(P_{\mathrm{req}}+\eta P_{\mathrm{aux}})}{-b_1\eta + b_2},  \frac{P_{\mathrm{em,max}}}{\eta}\\[6pt]
		&+ P_{\mathrm{aux}},\; \frac{T_{\mathrm{out,max}} - h_0 - h_1(P_{\mathrm{req}}+\eta P_{\mathrm{aux}})}{-h_1\eta + h_2},\;\\
		&\frac{T_{\mathrm{hpc,max}} - t_0 - t_1(P_{\mathrm{req}}+\eta P_{\mathrm{aux}})}{-t_1\eta + t_2}, \; \\
		&\frac{\dot{m}_{\mathrm{f,gt,min}} - e_0 - e_1(P_{\mathrm{req}}+\eta P_{\mathrm{aux}})}{-e_1\eta + e_2}
		\Bigg\},
	\end{aligned}
\end{equation} 
\begin{equation}
	\begin{aligned}
		P_{\mathrm{fc,min}}^{\mathrm{eff}} =&
		\max\Bigg\{
		P_{\mathrm{fc,min}},\;
		\frac{(P_{\mathrm{req}}+\eta P_{\mathrm{aux}} - P_{\mathrm{gt,max}})}{\eta},\; \\
		&\frac{\dot{m}_{\mathrm{f,fc,min}} - b_0 - b_1(P_{\mathrm{req}}+\eta P_{\mathrm{aux}})}{-b_1\eta + b_2}, \frac{P_{\mathrm{em,min}}}{\eta}\\[6pt]
		& + P_{\mathrm{aux}},\; \frac{T_{\mathrm{out,min}} - h_0 - h_1(P_{\mathrm{req}}+\eta P_{\mathrm{aux}})}{-h_1\eta + h_2},\; \\
		&\frac{T_{\mathrm{hpc,min}} - t_0 - t_1(P_{\mathrm{req}}+\eta P_{\mathrm{aux}})}{-t_1\eta + t_2}, \; \\
		&\frac{\dot{m}_{\mathrm{f,gt,max}} - e_0 - e_1(P_{\mathrm{req}}+\eta P_{\mathrm{aux}})}{-e_1\eta + e_2}
		\Bigg\}.
	\end{aligned}
\end{equation}
Therefore, the reduced constraint vector $\mathbf{g}^{\mathrm{eff}}$ contains the dominant limits as
		\begin{equation}
			\mathbf{g}^{\mathrm{eff}} = \begin{bmatrix}
				P_{\mathrm{fc,min}}^{\mathrm{eff}} - P_{\mathrm{fc}}\\
				P_{\mathrm{fc}} - P_{\mathrm{fc,max}}^{\mathrm{eff}}  \\ T_{\mathrm{in,min}}-T_{\mathrm{in}}(P_{\mathrm{fc}}) \\ 
				T_{\mathrm{in}}(P_{\mathrm{fc}}) - T_{\mathrm{in,max}} \\\dot{m}_{\mathrm{b,min}}-\dot{m}_{\mathrm{b}}(P_{\mathrm{fc}}) \\ 
				\dot{m}_{\mathrm{b}}(P_{\mathrm{fc}}) - \dot{m}_{\mathrm{b,max}}
			\end{bmatrix}.
		\end{equation}
		And the corresponding Lagrangian is written as
		\begin{equation} \label{eq:Ltoc}
			\begin{aligned}
				{L}=&\dot{m}_{\mathrm{f,fc}}+\dot{m}_{\mathrm{f,gt}}+\mu^{\mathrm{eff}}_1(P_{\mathrm{fc,min}}^{\mathrm{eff}} - P_{\mathrm{fc}}) + \\ 
				&\mu^{\mathrm{eff}}_2(P_{\mathrm{fc}} - P_{\mathrm{fc,max}}^{\mathrm{eff}}) + \mu^{\mathrm{eff}}_3(T_{\mathrm{in,min}}-T_{\mathrm{in}}) + \\ &\mu^{\mathrm{eff}}_4(T_{\mathrm{in}} - T_{\mathrm{in,max}})+\mu^{\mathrm{eff}}_{5}(\dot{m}_{\mathrm{b, min}} - \dot{m}_{\mathrm{b}}) + \\ 
				&\mu^{\mathrm{eff}}_{6}(\dot{m}_{\mathrm{b}} - \dot{m}_{\mathrm{b, max}}).
			\end{aligned}
		\end{equation}
		Applying the KKT stationarity condition
		\begin{equation}\label{eq:pL2toc}
			\frac{\mathrm{\partial}L}{\mathrm{\partial}P_{\mathrm{fc}}} = \frac{\mathrm{\partial}(\dot{m}_{\mathrm{f,fc}}+\dot{m}_{\mathrm{f,gt}})}{\mathrm{\partial}P_{\mathrm{fc}}} + \bm{\mu}^{{\mathrm{eff}}^\top} \cdot \frac{\mathrm{\partial}\mathbf{g}^{\mathrm{eff}}}{\mathrm{\partial}P_{\mathrm{fc}}},
		\end{equation}
		shows again that the dominant active constraint is the SOFC inlet temperature upper bound. Consequently, the corresponding optimal control policy is given by $g^{\mathrm{eff}}_4 = 0$, equals to  
		\begin{equation}\label{eq:pfctoc}
			P_{\mathrm{fc, const}}^* = \frac{-b + \sqrt{b^2-4ac}}{2a},
		\end{equation}
		where
		\begin{align}
			a &= Q_{22} - 2Q_{12}\eta+Q_{11}\eta^2, \\
			b &= 2(Q_{12}- \eta Q_{11})(P_{\mathrm{req}} + \eta P_{\mathrm{aux}}) -q_{1}\eta+q_{2},\\
			c &= Q_{11}(P_{\mathrm{req}}+\eta P_{\mathrm{aux}})^2 + q_{1}(P_{\mathrm{req}}+\eta P_{\mathrm{aux}}) + \notag\\
			&\quad q_0 - T_{\mathrm{in,max}}.
		\end{align}
\printcredits

\bibliographystyle{model1-num-names}

\bibliography{bibliography.bib}



\end{document}